\documentclass[a4paper,11pt,oneside]{article}

\usepackage[left=1.0in, right=1.0in, top=1.0in, bottom=1.0in]{geometry}

\usepackage[toc, page]{appendix}
\usepackage{amsmath}
\usepackage{amsfonts}
\usepackage{amssymb}
\usepackage{amsthm}
\usepackage{mathtools}
\usepackage{color}
\usepackage{graphics}
\usepackage{graphicx}
\usepackage{epsfig}
\usepackage{epstopdf}
\usepackage{float}
\usepackage{rotating}
\usepackage{array}
\usepackage{textcomp}
\usepackage{bbm}
\usepackage{fancyref}
\usepackage{hyperref}
\usepackage{xcolor}
\usepackage{rotating}
\usepackage[noadjust]{cite} % basic fos bibliography [1]-[5]
\usepackage{tikz}
\usetikzlibrary{calc,positioning,shapes,shadows,arrows,fit}
\usepackage{algorithm}
\usepackage{algorithmic}
\usepackage{wrapfig}

\newtheorem{theorem}{Theorem}
\newtheorem{lemma}{Lemma}

\newtheorem{proposition}{Proposition}

\theoremstyle{definition}
\newtheorem{definition}{Definition}
\newtheorem{remark}{Remark}
\newtheorem{assumption}{Assumption}

\newtheorem{problem}{Problem}
\newtheorem{property}{Property}

\title{Decentralized Tube-based Model Predictive Control of Uncertain Nonlinear Multi-Agent Systems}
		
\author{Alexandros Nikou and Dimos V. Dimarogonas%
\thanks{The authors are with the School of Electrical Engineering and Computer Science, KTH Royal Institute of Technology, SE-100 44, Stockholm, Sweden and with the KTH Center for Autonomous Systems. Email: {\tt\small \{anikou, dimos\}@kth.se}. This work was supported by the H2020 ERC Starting Grant BUCOPHSYS, the EU H2020 Co4Robots project, the Swedish Foundation for Strategic Research (SSF), the Swedish Research Council (VR) and the Knut och Alice Wallenberg Foundation (KAW).}}

\date{}

\begin{document}
\maketitle

\begin{abstract}
This paper addresses the problem of decentralized tube-based nonlinear Model Predictive Control (NMPC) for a class of uncertain nonlinear continuous-time multi-agent systems with additive and bounded disturbance. In particular, the problem of robust navigation of a multi-agent system to predefined states of the workspace while using only local information is addressed, under certain distance and control input constraints. We propose a decentralized feedback control protocol that consists of two terms: a nominal control input, which is computed online and is the outcome of a Decentralized Finite Horizon Optimal Control Problem (DFHOCP) that each agent solves at every sampling time, for its nominal system dynamics; and an additive state feedback law which is computed offline and guarantees that the real trajectories of each agent will belong to a hyper-tube centered along the nominal trajectory, for all times. The volume of the hyper-tube depends on the upper bound of the disturbances as well as the bounds of the derivatives of the dynamics. In addition, by introducing certain distance constraints, the proposed scheme guarantees that the initially connected agents remain connected for all times. Under standard assumptions that arise in nominal NMPC schemes, controllability assumptions as well as communication capabilities between the agents, we guarantee that the multi-agent system is ISS (Input to State Stable) with respect to the disturbances, for all initial conditions satisfying the state constraints. Simulation results verify the correctness of the proposed framework. \\

\noindent \textbf{Keywords :} Nonlinear Model Predictive Control (NMPC), Robust Control, Tube-based MPC, Multi-Agent Systems, Navigation, Connectivity Maintenance, Input to State Stability (ISS).
\end{abstract}

\section{Introduction}

During the last decades, \emph{NMPC} has been proven to be a powerful control framework for dealing with the problem of stabilization of dynamical systems under state and input constraints \cite{michalska_1993, frank_1998_quasi_infinite, mayne_2000_nmpc, frank_2003_nmpc_bible, morrari_npmpc, grune_2011_nonlinear_mpc, borrelli_2013_nmpc}. One of the main challenges in NMPC is the treatment of potential uncertainties due to imperfect modeling and/or disturbances that may affect the system. In parallel to that, \emph{decentralized control of multi-agent systems} has gained significant attention due to its great variety of applications including multi-robot systems, transportation and biological systems \cite{olfati_murray_concensus, ren_beard_concensus}. In this paper, we aim to exploit a novel robust MPC framework in order to solve a collaborative multi-agent navigation problem under certain distance and input constraints.

The literature on the problem of robust NMPC has been extensively in the last years. Authors in \cite{camacho_2002_input_to_state, parisini_2009_restricted_sets} proposed a method of constraint sets tightening for guaranteeing robust stability. A  min-max robust MPC approach has been provided in \cite{maciejowski_minmax, camacho_2006_minmax, lazar_min_max}. A promising robust strategy, originally proposed for discrete-time linear systems in \cite{rakovic_2004_tubes_1, rakovic_2005_tubes_2, gonzalez2011online}, is the so called tube-based approach. In tube-based MPC a Finite Horizon Optimal Control Problem (FHOCP) is solved online for the nominal system, while the real trajectory is guaranteed to remain in a bounded tube for all times. Tube-based approaches for nonlinear discrete-time systems have been considered in \cite{mayne2007tube, magni2003robust, cannon2011_tubes, mayne2011tube, bayer2013discrete}. Authors in \cite{farina2012tube} have addressed the linear continuous-time case. In \cite{china1, china2}, regarding the computation of the offline feedback controller, the discrepancy between the nominal nonlinear system with the corresponding linear system has been considered. In \cite{yu_2013_tube}, sufficient conditions for affine in the control continuous-time nonlinear systems with constant matrices multiplying the control input vectors have been proposed, which we aim to extend here in order to cover a larger class of nonlinear systems, and in particular decentralized multi-agent systems.

\emph{Multi-agent navigation} in an important field in both the robotics and the control communities, due to the need for autonomous control
of multiple robotic agents in the same workspace \cite{dimos_2006_automatica_nf, panagou_potential_fields_unicycle, loizou_2017_navigation_transformation}. Important applications of multi-agent navigation arise also in the fields of air-traffic management and autonomous driving for guaranteeing collision avoidance with other cars and obstacles. In this work, we study the problem of decentralized navigation for nonlinear multi-agent systems with network connectivity maintenance. Decentralized NMPC schemes for multi-agent navigation have been investigated in \cite{distributed_mpc_sastry, fukushima_2005_distributed_mpc_linearization, trodden2006robust, richards2007robust, Dunbar2006549, aguiar_multiple_uavs_linearized_mpc, IJC_2018_alex}. In the latter works, however, either the under consideration dynamics are linear/simple or the robust MPC frameworks do not invoke the tube-based online optimization and offline feedback control design approach. Other applications are formation control, in which the agents are required to reach a predefined geometrical shape (see e.g., \cite{alex_chris_automatica_2018}) and high-level planning where it is required to provide decentralized controllers for navigating the agents between regions of interest of the workspace (see e.g., \cite{alex_acc_2016}).

Motivated by the aforementioned, the contribution of the paper is to propose a decentralized tube-based feedback control protocol for a general class of uncertain nonlinear continuous-time multi-agent systems, i.e., a more general class of systems than the ones that have already been studied in the literature. More specifically, each agent solves a DFHOCP and exchanges its open-loop predicted trajectory with its neighbors in order to reach a predefined state, under certain distance and inputs constraints. The proposed control consists of two parts: the first part is the online solution to a nominal DFHOCP, which is solved at every constant sampling time; the second part is a state feedback law which is calculated offline in order to guarantee that the error between the real and the nominal trajectory remains in a bounded hyper-tube, for all times. Under standard nominal NMPC and communication capabilities assumptions (see \cite{frank_1998_quasi_infinite, muller2012cooperative}) between neighboring agents, we show that the proposed control law renders the closed multi-agent system ISS with respect to the disturbances. The contributions of this paper are summarized as follows:
\begin{itemize}
	\item We develop a systematic control design methodology for tube-based NMPC which guarantees ISS for uncertain nonlinear non-affine continuous-time systems.
	\item The above results are exploited and extended for solving a constrained navigation multi-agent problem under coupled constraints in a decentralized manner.
\end{itemize}

The remainder of this manuscript is organized as follows: In Section \ref{sec:notation_preliminaries} the notation and preliminaries are given. Section \ref{sec:problem_formulation} provides the system dynamics under consideration and the problem statement. Section \ref{sec:main_results} discusses the technical details of the proposed solution and Section \ref{sec:simulation_results} is devoted to a simulation example. Finally, conclusions and future work are discussed in Section \ref{sec:conclusions}.

\section{Notation and Preliminaries} \label{sec:notation_preliminaries}

The sets of positive integers and real numbers are denoted by $\mathbb{N}$ and $\mathbb{R}$, respectively; $\mathbb{R}^n_{\geq 0}$ and $\mathbb{R}^n_{> 0}$ are the sets of real
$n$-vectors with all elements nonnegative and positive, respectively. Sets, vectors and matrices will be denoted by calligraphic, small and capital letters, respectively. The notation $\|x\| \coloneqq \sqrt{x^\top x}$ is used for the Euclidean norm of a vector $x \in \mathbb{R}^n$; $I_n \in \mathbb{R}^{n \times n}$ and $0_{m \times n} \in \mathbb{R}^{m \times n}$ are the identity matrix and the $m \times n$ matrix with all entries zeros, respectively; $\lambda_{\min}(A)$ denotes the minimum absolute value of the real part of the eigenvalues of a matrix $A \in \mathbb{R}^{n \times n}$. The set-valued function $\mathcal{B}:\mathbb{R}^n\times\mathbb{R}_{> 0} \rightrightarrows \mathbb{R}^n$, defined by $\mathcal{B}(x, r) \coloneqq \{y \in \mathbb{R}^n: \|y-x\| \leq r\},$ represents the $n$-th dimensional ball with center $x \in \mathbb{R}^{n}$ and radius $r \in \mathbb{R}_{> 0}$. Given a vector valued function $f: \mathbb{R}^n \to \mathbb{R}^m$, $\frac{\partial f_i}{\partial x_j}$ denotes the element of row $i$ and column $j$ of the Jacobian matrix of $f$, with $i$, $j \in \{1,\dots,n\}$. Define a vector of a canonical basis of $\mathbb{R}^n$ by:
\begin{align} \label{eq:ell_vectors}
\ell_n(i) \coloneqq \left[0, \dots, 0, \underbrace{1}_{i-\rm{th \ element}}, 0, \dots, 0 \right]^\top \in \mathbb{R}^n.
\end{align}

\begin{definition}
	Given two vectors $x$, $y \in \mathbb{R}^n$ their \emph{convex hull} is defined by:
	\begin{align*} \label{eq:convex_hull}
	\rm{Co}(x,y) \coloneqq \{\xi : \xi = \theta x + (1-\theta) y, 0 < \theta < 1 \}.
	\end{align*}
\end{definition}

\begin{definition} \label{def:p_difference}
	Given the sets $\mathcal{S}_1$, $\mathcal{S}_2 \subseteq \mathbb{R}^n$ and the matrix $A \in \mathbb{R}^{n \times m}$, the \emph{Minkowski addition}, the \emph{Pontryagin difference} and the \emph{matrix-set multiplication} are respectively defined by: 
	\begin{align*}
	\mathcal{S}_1 \oplus \mathcal{S}_2 & \coloneqq \{s_1 + s_2 \in \mathbb{R}^n : s_1 \in \mathcal{S}_1, s_2 \in \mathcal{S}_2\}, \notag \\
	\mathcal{S}_1 \ominus \mathcal{S}_2 & \coloneqq \{s_1 \in \mathbb{R}^n: s_1+s_2 \in \mathcal{S}_1, \forall s_2 \in \mathcal{S}_2\}, \notag \\
	A \circ \mathcal{S} & \coloneqq \{a: \exists s \in \mathcal{S}, a = As\}.
	\end{align*}
\end{definition}

\begin{property} (\cite{kolmanovsky1998theory})\label{prop:set_prop}
	Let $\mathcal{S}_1$, $\mathcal{S}_2 \subseteq \mathbb{R}^n$ and assume that $\mathcal{S}_1 \ominus \mathcal{S}_2 \neq \emptyset$. Then it holds that $\big(\mathcal{S}_1 \ominus \mathcal{S}_2 \big) \oplus \mathcal{S}_2 \subseteq \mathcal{S}_1$.
\end{property}

\begin{lemma} \label{lemma:quadr_forms}
	For any vectors $x, y \in \mathbb{R}^{n}$, positive definite matrix $M \in \mathbb{R}^{n \times n}$ and constant $\rho > 0$ it holds that:
	\begin{equation} \label{eq:basic_ineq}
	x^\top M y \le \frac{1}{4 \rho} x^\top M x + \rho y ^\top M y.
	\end{equation}
\end{lemma}
\begin{proof}
	By using the facts that $\rho >0$, $M > 0$ the following equivalences hold
	$(x-2\rho y)^\top M (x- 2\rho y) \ge 0$ $\Leftrightarrow x^\top M x - 2 \rho x^\top M y - 2 \rho y^\top M x + 4 \rho^2 y^\top M y \ge 0$ $\Leftrightarrow x^\top M y + x^\top M^\top y \le \frac{1}{2 \rho} x^\top M x + 2 \rho y^\top M y$ $\Leftrightarrow x^\top M y \le \frac{1}{4 \rho} x^\top M x + \rho y^\top M y$.
\end{proof}

\begin{proposition} \label{eq:MVT}
	$[$Mean Value Theorem for vector valued functions$]$ \cite{zemouche2008observers} Consider a vector valued function $f: \mathbb{R}^n \to \mathbb{R}^m$. Assume that $f$ is differentiable on an open set $\mathcal{S} \subseteq \mathbb{R}^n$. Let $x$, $y$ two points of $\mathcal{S}$ such that ${\rm Co}(x,y) \subseteq \mathcal{S}$. Then, there exist contanst vectors $\xi_1$, $\dots,$ $\xi_m \in {\rm Co}(x,y)$ such that:
	\begin{align} \label{eq:eq_MVT}
	f(x)-f(y) = \left[ \sum_{k = 1}^{m} \sum_{j=1}^{n} \ell_m(k) \ell_n(j)^\top \frac{\partial f_k(\xi_k)}{\partial x_j} \right] (x-y).
	\end{align}
\end{proposition}

\begin{definition} \label{def:k_class} \cite{khalil_nonlinear_systems} A continuous function $\alpha : [0, a) \to \mathbb{R}_{\ge 0} $ belongs to \emph{class $\mathcal{K}$} if it is strictly increasing and $\alpha (0) = 0$. A continuous function $\beta : [0, a) \times \mathbb{R}_{\ge 0} \to \mathbb{R}_{\ge 0}$ belongs to \emph{class $\mathcal{KL}$} if: $1)$ for a fixed $s$, the mapping $\beta(r,s)$ belongs to class $\mathcal{K}$ with respect to $r$; $2)$ for a fixed $r$, the mapping $\beta(r,s)$ is strictly decreasing with respect to $s$; and it holds that $\lim\limits_{s \to \infty} \beta(r,s) = 0$.
\end{definition}

\begin{definition} \label{def:RPI_set}
	Consider a system $\dot{x} = f(x,u,w)$ where: $x \in \mathcal{X}$, $u \in \mathcal{U}$, $w \in \mathcal{W}$ with initial condition $x(0) \in \mathcal{X}$. A set $\mathcal{S} \subseteq \mathcal{X}$ is a \emph{Robust Control Invariant (RCI) set} for the system, if there exists a feedback control law $u \coloneqq \kappa(x) \in \mathcal{U}$, such that for all $x(0) \in \mathcal{S}$ and for all $w(t) \in \mathcal{W}$ it holds that $x(t) \in \mathcal{S}$ for all $t \in \mathbb{R}_{\ge 0}$, along every solution $x(t)$.
\end{definition}

\begin{theorem} \label{teheorem:uub_theorem} \cite{khalil_nonlinear_systems}
	Consider a nonlinear system $\dot{x} = f(x)$ where $x \in \mathcal{X}$. Let $V:\mathcal{X} \to \mathbb{R}$ be a continuously differentiable function such that $\alpha_1(\|x\|) \le V(x) \le \alpha_2(\|x\|)$ and $\dot{V}(x) \le - g(x)$, $\forall \|x\| \ge \xi > 0$, where $\alpha_1$, $\alpha_2$ are class $\mathcal{K}$ functions and $g$ is a positive definite function. Then, there exists a finite time $\tau > 0$ such that the solution $x(t)$ satisfies $\|x(t)\| \le \alpha_{1}^{-1}(\alpha_2(\xi))$, $\forall t \ge \tau$.
\end{theorem}

\begin{definition} \cite{khalil_nonlinear_systems} \label{def:ISS}
	Consider a nonlinear system $\dot{x} = f(x,u,w)$, where: $x \in \mathcal{X}$, $u \in \mathcal{U}$, $w \in \mathcal{W}$ with initial condition $x(0) \in \mathcal{X}$. The system is said to be \emph{Input-to-State Stable (ISS)} with respect to $w \in \mathcal{W}$, if there exist functions $\beta \in \mathcal{KL}$, $\gamma \in \mathcal{K}$ such that for any initial condition $x(0) \in \mathcal{X}$ and any bounded input $u \coloneqq \kappa(x) \in \mathcal{U}$, the solution $x(t)$ exists for all $t \in \mathbb{R}_{\ge 0}$ and satisfies: $$\|x(t)\| \leq \beta\big(\|x(0)\|,t\big) + \gamma \left(\displaystyle \sup_{0 \le s \le t} \|w(s)\|\right).$$
\end{definition}

\begin{definition} \cite{yu_2013_tube} \label{def:asympt_ultimately_bounded}
	A nonlinear system is \emph{asymptotically ultimately bounded} if a set of initial conditions of the system converges to a bounded set.
\end{definition}

\section{Problem Formulation} \label{sec:problem_formulation}

\subsection{System Model}

\noindent Consider a set $\mathcal{V}$ of $N$ agents, $\mathcal{V} = \{1, \dots, N\}$, operating in a workspace $\mathcal{D} \subseteq \mathbb{R}^{n}$; $\mathcal{D}$ is assumed to be a \emph{connected} set containing the origin. The \emph{uncertain nonlinear continuous dynamics} of each agent $i \in \mathcal{V}$ are given by:
\begin{align} 
\dot{x}_{i}(t) & =  f_i(x_i(t), u_i(t))+w_i(t), \label{eq:sys}
\end{align}
where $x_i(t) \in \mathcal{D}$ denotes the state of each agent; $u_i(t) \in \mathbb{R}^n$ denotes the control input; $f_i:\mathbb{R}^n \times \mathbb{R}^n \to \mathbb{R}^n$ is a continuous nonlinear vector valued function; and the term $w_i(t) \in \mathbb{R}^{n}$ represents external disturbances/uncertainties as well as unmodeled dynamics. The state $x_i(t)$ of each agent $i \in \mathcal{V}$ is assumed to be available for measurement for all times. The control inputs are assumed to satisfy: $u_i(t) \in \mathcal{U}_i \subseteq \mathbb{R}^n$, for every $t \in \mathbb{R}_{\ge 0}$, where $\mathcal{U}_i$ are convex sets containing the origin. Denote by $x_i(0) \in \mathcal{D}$ the initial condition of \eqref{eq:sys}. Assume also that the uncertainties are bounded i.e., there exists finite constants $\widetilde{w}_i$ such that:
\begin{equation} \label{eq:widetilde_w}
w_i(t) \in \mathcal{W}_i \coloneqq \{w_i \in \mathbb{R}^n: \|w_i\| \le \widetilde{w}_i \}, \forall t \in \mathbb{R}_{\ge 0}, i \in \mathcal{V}.
\end{equation}
For system \eqref{eq:sys}, define the \emph{nominal system} (without disturbances) by:
\begin{align} 
\dot{\overline{x}}_{i}(t) & =  f_i(\overline{x}_i(t), \overline{u}_i(t)), \label{eq:nominal_system}
\end{align}
where $w_i(t) = 0$, $\overline{x}_i(t) \in \mathcal{D}$ and $\overline{u}_i(t) \in \mathcal{U}_i$, for every $t \in \mathbb{R}_{\ge 0}$, $i \in \mathcal{V}$. Hereafter, we shall denote by $\overline{\cdot}$ all the nominal signals.

\begin{assumption} \label{ass:f_assumption}
	The nonlinear functions $f_i: \mathcal{D} \times \mathcal{U}_i \to \mathbb{R}^n$ are \emph{continuously differentiable} with respect both to $x_i$ and $u_i$ in $\mathcal{D} \times \mathcal{U}_i$ with $f_i(0,0) = 0$, $\forall i \in \mathcal{V}$.
\end{assumption}

\begin{assumption} \label{ass:stabilized_linear_assumption}
	The linear systems $\dot{\overline{x}}_i(t) = A_i \overline{x}_i(t) + B_i \overline{u}_i(t)$, that are the outcome of the linearization of the nominal systems \eqref{eq:nominal_system} around the equilibrium states $x_i = 0$ are stabilizable.	
\end{assumption}

\noindent Define the function $J: \mathcal{D} \times \mathcal{U}_i \to \mathbb{R}^n$ by:
\begin{equation} \label{eq:function_J_i}
J_i(x_i, u_i) \coloneqq \sum_{k = 1}^{n} \sum_{j=1}^{n} \ell_n(k) \ell_n(j)^\top \frac{\partial f_{i,k}(x_i, u_i)}{\partial u_j},
\end{equation}
where $f_{i,k}$ stands for the $k$-th component of the vector valued function $f_i$ and the vectors $\ell_n(\cdot)$ as defined in \eqref{eq:ell_vectors}.

\begin{assumption} \label{ass:lower_bound_deriv}
	We assume that there exist constants $\underline{J}_i$ such that:
	\begin{align} \label{eq:lambda_min}
	\lambda_{\min}\left[\frac{J_i(x_i, u_i) + J^\top_i(x_i, u_i)}{2}\right] \ge \underline{J}_i > 0, \forall x_i \in \mathcal{D}, u_i \in \mathcal{U}_i.
	\end{align}
\end{assumption}

\begin{remark}
	Assumptions \ref{ass:f_assumption}, \ref{ass:stabilized_linear_assumption} are standard assumptions required for the NMPC nominal stability to be guaranteed (see \cite{frank_1998_quasi_infinite}). Assumption \ref{ass:lower_bound_deriv} is a sufficient controllability condition for nonlinear systems in non-affine form (see e.g., \cite{adaptive_non_affine}).
\end{remark}

\subsection{Objectives}

Given the aforementioned modeling, the objective of each agent $i \in \mathcal{V}$ is to reach to a pre-defined desired configuration $x_{\scriptscriptstyle i, \rm des} \in \mathcal{D}$ of the workspace from any initial conditions $x_i(0) \in \mathcal{D}$. Moreover, motivated by practical applications in which agents need to stay sufficiently close in order to execute collaborative tasks, it is desired to introduce connectivity maintenance coupled constraints between the agents. For this reason, assume that:
\begin{itemize}
	\item over time $t \in \mathbb{R}_{\ge 0}$, each agent $i \in \mathcal{V}$ occupies a ball $\mathcal{B}(x_i(t), r_i)$, where $r_i \in \mathbb{R}_{>0}$ stands for the radius of the ball;
	\item each agent $i \in \mathcal{V}$ has \emph{communication capabilities} within a limited sensing range $d_i \in \mathbb{R}_{>0}$ such that:
	\begin{equation} \label{eq:d_i}
	\displaystyle d_i > \max_{i, j \in \mathcal{V}, i \neq j} \left\{r_i + r_j \right\}.
	\end{equation}
\end{itemize}
The latter implies that each agent has sufficiently large sensing radius so as to measure the agent with the biggest volume in its vicinity, due to the fact that the agents' radii are not the same. Taking the above into consideration, the neighboring set of agent $i \in \mathcal{V}$ is defined by:
\begin{equation*}
\mathcal{N}_i \coloneqq \left\{j \in \mathcal{V} \backslash \{i\} : \| x_i(0)-x_j(0) \| < d_i \right\}.
\end{equation*}
The set $\mathcal{N}_i$ is composed of indices of agents $j \in \mathcal{N}$ which are within the sensing range of agent $i$ at time $t = 0$. The proposed decentralized feedback control laws $u_i \in \mathcal{U}_i$ need to guarantee that all agents $j \neq i$ within $\mathcal{N}_i$ remain within distance $d_i$ from $i$ att all times, i.e., the connectivity of all initially connected agents is preserved for all times. For sake of cooperation needs, assume that $\mathcal{N}_i \neq \emptyset$, $\forall i \in \mathcal{V}$, i.e., all agents have at least one neighbor. Moreover, assume that the workspace $\mathcal{D}$ is sufficiently large so that the agents are able to perform the desired navigation task.

\begin{definition} \label{definition:feasibl_configurations}
	The desired configurations $x_{\scriptscriptstyle i, \rm des} \in \mathcal{D}$, $i \in \mathcal{V}$ are called \emph{feasible}, if they do not result in violation of the connectivity maintenance constraint between the agents, i.e., $\|x_{\scriptscriptstyle i, \rm des} - x_{\scriptscriptstyle j, \rm des}\| < d_i$, $\forall i \in \mathcal{V}$, $j \in \mathcal{N}_i$.
\end{definition}

\subsection{Problem Statement}

\noindent Formally, the control problem considered in this paper, is formulated as follows:

\begin{problem} \label{problem}
	Given $N$ agents governed by dynamics as in \eqref{eq:sys}, under Assumptions \ref{ass:f_assumption}- \ref{ass:lower_bound_deriv}, modeled by the balls $\mathcal{B}\left(x_i, r_i\right)$, $i \in \mathcal{V}$, and operating in a workspace $\mathcal{D}$. The agents have communication capabilities according to sensing radii $d_i$, as given in \eqref{eq:d_i}. Then, given desired feasible configurations $x_{\scriptscriptstyle i, \rm des} \in \mathcal{D}$, $i \in \mathcal{V}$ according to Definition \ref{definition:feasibl_configurations}, the problem lies in designing \emph{decentralized feedback control} laws $u_i \in \mathcal{U}_i$, such that for every $i \in \mathcal{V}$ and for all initial conditions satisfying $x_i(0) \in \mathcal{D}$ the following specifications are satisfied: 
	\begin{enumerate}
		\item navigation to the desired configurations is achieved: $\displaystyle \lim_{t \to \infty} \|x_i(t) - x_{\scriptscriptstyle i, \rm des} \| \to 0;$
		\item connectivity between neighboring agents is preserved: $\|x_i(t) - x_j(t)\| < d_i$, $\forall j \in \mathcal{N}_i$, $t \in \mathbb{R}_{\ge 0}$; and
		\item the agents remain in the workspace: $x_i(t) \in \mathcal{D}$, $\forall t \in \mathbb{R}_{\ge 0}$.
	\end{enumerate}
\end{problem}

\section{Main Results} \label{sec:main_results}

In this section, a systematic solution to Problem \ref{problem} is introduced. Due to the fact that we aim to minimize the terms $\|x_i(t) - x_{\scriptscriptstyle i, \rm des} \|$, as $t \to \infty$, subject to distance and control input constraints imposed by Problem \ref{problem}, we seek a solution which is the outcome of an decentralized optimization. In Section \ref{sec:error_dynamics} we derive the error dynamics and the distance constraints of each agent; Section \ref{sec:dec_control_design} is devoted to the proposed feedback control design; and lastly in Section \ref{sec:online_control_design} we deal with the online nominal NMPC design.

\subsection{Error Dynamics and Constraints} \label{sec:error_dynamics}

\noindent Define the uncertain error and nominal error signals $e_i: \mathbb{R}_{\ge 0} \to \mathbb{R}^n$, $\overline{e}_i : \mathbb{R}_{\ge 0} \to \mathbb{R}^n$ by:
\begin{align*}
e_i(t) & \coloneqq x_i(t)-x_{\scriptscriptstyle i, \rm des}, \notag \\
\overline{e}_i(t) & \coloneqq \overline{x}_i(t)-x_{\scriptscriptstyle i, \rm des},
\end{align*}	
respectively. Then, the corresponding \emph{uncertain error dynamics} and \emph{nominal error dynamics} are given by:
\begin{subequations}
	\begin{align}
	\dot{e}_{i}(t) & = f_i(e_i(t)+x_{\scriptscriptstyle i, \rm des}, u_i(t))+w_i(t),  \label{eq:error_system_perturbed} \\
	\dot{\overline{e}}_{i}(t) & =  f_i(\overline{e}_i(t)+x_{\scriptscriptstyle i, \rm des}, \overline{u}_i(t)). \label{eq:error_system_nominal}
	\end{align}
\end{subequations}
Define the sets that captures the \textit{state} constraints on the system \eqref{eq:sys}, posed by Problem \ref{problem} by:
\begin{align} \label{eq:constraints_X}
\mathcal{X}_{i} & \coloneqq \Big\{ x_i(t) \in \mathbb{R}^n: \ \|x_i(t) - x_j(t)\| \le d_i - \varepsilon, \forall \ j \in \mathcal{V}_i \ {\rm and} \ x_i(t) \in \mathcal{D}\Big\}, i \in \mathcal{V}, 
\end{align}
where $\varepsilon \in \mathbb{R}_{> 0}$ is an arbitrary small constant to be chosen. The two constraints in the set $\mathcal{X}_i$ refer to connectivity preservation between neighboring agents and the requirement of the agents to remain in the workspace $\mathcal{D}$ for all times, respectively. In order to translate the constraints that are dictated for the state $x_i$ into constraints regarding the error state $e_i$, define the set $$\mathcal{E}_{i} = \Big\{e_i \in \mathbb{R}^n :
e_i \in \mathcal{X}_{i} \oplus (-x_{\scriptscriptstyle i, \rm des}) \Big\}, i \in \mathcal{V}.$$ Then, the following equivalence holds: $x_i \in \mathcal{X}_i \Leftrightarrow e_i \in \mathcal{E}_i$, $\forall i \in \mathcal{V}$.	

\subsection{Feedback Control Design} \label{sec:dec_control_design}

\noindent Consider the following \emph{feedback control law:}
\begin{equation} \label{eq:control_law_u_kappa}
u_i(t) \coloneqq \overline{u}_i(t) + \kappa_i(e_i(t), \overline{e}_i(t)),
\end{equation}
which consists of a nominal control action $\overline{u}_i \in \mathcal{U}_i$ and a state feedback law $\kappa_i: \mathbb{R}^n \times \mathbb{R}^n \to \mathbb{R}^n$. As it will be presented hereafter, $\overline{u}_i$ will be the outcome of a nominal DFHOCP solved at each sampling time by each agent $i \in \mathcal{V}$; and the feedback law $\kappa_i(e_i, \overline{e}_i)$ is used to guarantee that the real trajectories $e_i(t)$ remain in bounded hyper-tubes centered among the nominal trajectories $\overline{e}_i(t)$, for all times. We will show that the volume of the hyper-tubes depends on the upper bound of the disturbances $\widetilde{w}_i$ as given in \eqref{eq:widetilde_w} as well as the bounds of the derivatives of functions $f_i$. For each agent $i \in \mathcal{V}$, denote by:
\begin{align} 
z_{i}(t) & \coloneqq e_i(t) - \overline{e}_i(t), \label{eq:error_z}
\end{align}
the deviation between the real states $e_i(t)$ of the uncertain system \eqref{eq:error_system_perturbed} and the states $\overline{e}_i(t)$ of the nominal system \eqref{eq:error_system_nominal}. Note that initially it holds that $z_{i}(0) = 0$ for every $i \in \mathcal{V}$. By using \eqref{eq:error_system_perturbed}, \eqref{eq:error_system_nominal} and \eqref{eq:error_z}, the dynamics of $z_{i}(t)$ for each agent $i \in \mathcal{V}$ are given by:
\begin{align}
\dot{z}_{i} & = \dot{e}_i - \dot{\overline{e}}_i, \notag \\
& = f_i(e_i+x_{\scriptscriptstyle i, \rm des}, u_i)+w_i - f_i(\overline{e}_i+x_{\scriptscriptstyle i, \rm des}, \overline{u}_i) \notag \\
& = f_i(e_i+x_{\scriptscriptstyle i, \rm des}, u_i) - f_i(\overline{e}_i+x_{\scriptscriptstyle i, \rm des}, u_i) +f_i(\overline{e}_i+x_{\scriptscriptstyle i, \rm des}, u_i)- f_i(\overline{e}_i+x_{\scriptscriptstyle i, \rm des}, \overline{u}_i)+w_i \notag \\
& = \Lambda_i(e_i, \overline{e}_i, u_i)+f_i(\overline{e}_i+x_{\scriptscriptstyle i, \rm des}, u_i)- f_i(\overline{e}_i+x_{\scriptscriptstyle i, \rm des}, \overline{u}_i)+w_i. \label{eq:z_dynamics}
\end{align}
In the latter, the functions $\Lambda_i:\mathcal{D} \times \mathcal{D} \times \mathcal{U}_i \to \mathbb{R}^{6}$ are defined by:
\begin{align*} 
\Lambda_i(e_i, \overline{e}_i, u_i) & \coloneqq f_i(e_i+x_{\scriptscriptstyle i, \rm des}, u_i) - f_i(\overline{e}_i+x_{\scriptscriptstyle i, \rm des}, u_i),
\end{align*}
and are upper bounded by:
\begin{align} \label{eq:L_bound}
\|\Lambda_i(e_i, \overline{e}_i, u_i)\| & \le \|f_i(e_i+x_{\scriptscriptstyle i, \rm des}, v_i, u_i) - f_i(\overline{e}_i+x_{\scriptscriptstyle i, \rm des}, v_i, u_i)\|  \notag \\
& \le L_{i} \|e_i+x_{\scriptscriptstyle i, \rm des}-\overline{e}_i -x_{\scriptscriptstyle i, \rm des}\| \notag \\
& = L_{i} \|e_i-\overline{e}_i\| \notag \\
& = L_{i} \|z_{i}\|,
\end{align}
where $L_{i} > 0$ are the Lipschitz constants of the functions $f_i$ with respect to the variables $x_i$.
\begin{figure}[t!]
	\centering
	\begin{tikzpicture}
	% draw the ellipses
	\draw[color = green, fill=black!5] (-0.55,-0.48) ellipse (0.35cm and 1.97cm);
	\draw[color = green, fill=black!5] (2.55,-0.15) ellipse (0.35cm and 1.97cm);
	\draw[color = green, fill=black!5] (6.25,-0.22) ellipse (0.35cm and 1.97cm);
	
	% draw the nominal trajectory
	\draw[scale=0.7,domain=-0.8:9,smooth,variable=\x,blue, line width = 0.05cm] plot ({\x},{-0.005*\x*\x*\x*\x+0.0868*\x*\x*\x-0.4554*\x*\x+0.6542*\x+0.1467});
	
	% draw the real trajectory 
	\draw[scale=0.7,domain=-0.8:9,smooth,variable=\x,red, line width = 0.05cm] plot ({\x},{0.93+0.0139*\x*\x*\x*\x-0.2491*\x*\x*\x+1.2240*\x*\x-0.7324*\x-3.0411});
	
	% draw the hyper-tube
	\draw[scale=0.7,domain=-0.8:9,smooth,variable=\x,black, line width = 0.05cm, dashed] plot ({\x},{-0.005*\x*\x*\x*\x+0.0868*\x*\x*\x-0.4554*\x*\x+0.6542*\x+0.1467+2.8});
	\draw[scale=0.7,domain=-0.8:9,smooth,variable=\x,black, line width = 0.05cm, dashed] plot ({\x},{-0.005*\x*\x*\x*\x+0.0868*\x*\x*\x-0.4554*\x*\x+0.6542*\x+0.1467-2.8});
	
	%\draw (-0.53, -0.48) node[circle, inner sep=0.8pt, fill=black, label={below:{$\overline{e}_i(0)$}}] (A1) {};
	\draw (6.29, -0.3) node[circle, inner sep=0.8pt, fill=black, label={below:{$ $}}] (A1) {};
	
	% draw the hyper-tube's diameter
	\draw [color=orange,thick,->,>=stealth', line width = 0.5mm](-1.5, -0.5) to (-1.5, 1.4);
	\draw [color=orange,thick,->,>=stealth', line width = 0.5mm](-1.5, 1.4) to (-1.5, -0.5);
	\draw [color=orange,thick,->,>=stealth', line width = 0.5mm](7.0, -0.3) to (7.0, -2.0);
	\draw [color=orange,thick,->,>=stealth', line width = 0.5mm](7.0, -2.0) to (7.0, -0.3);
	\node at (-1.8, 0.5) {$\widetilde{z}_i$};
	\node at (7.4, -1.2) {$\widetilde{z}_i$};
	\node at (2.6,-0.16) {$\bullet$};
	\node at (6.3,-0.27) {$\bullet$};
	\node at (-0.5,-0.50) {$\bullet$};
	%\node at (-0.52,0.70) {$\bullet$};
	
	%\node at (-0.0, 1.0) {$e_i(0)$};
	\node at (3.9, 0.9) {$e_i(t)$};
	\node at (0.6, 0.7) {$\overline{e}_i(t)$};
	\end{tikzpicture}
	\caption{The hyper-tube of agent $i$ centered along the trajectory $\overline{e}_i(t)$ (depicted by blue line) with radius $\widetilde{z}_i$. Under the proposed control law, the real trajectory $e_i(t)$ (depicted with red line) lies inside the hyper-tube for all times, i.e., $\|z_{i}(t)\| \le \widetilde{z}_i$, $\forall t \in \mathbb{R}_{\ge 0}$.}
	\label{fig: example_01}
\end{figure}
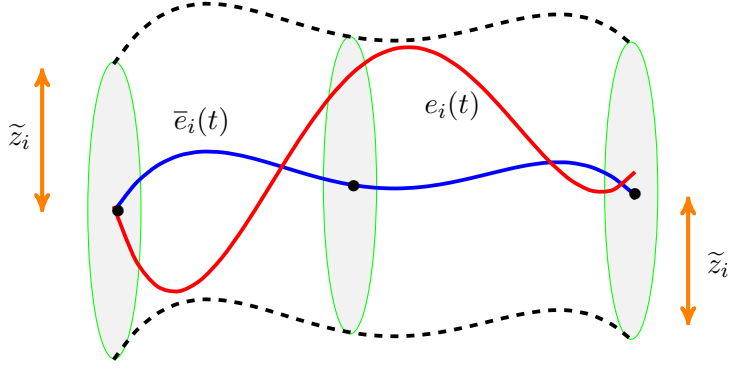
\begin{lemma} \label{leamm:proposed_RPI_set}
	The state feedback laws designed as:
	\begin{equation} \label{eq:kappa_law}
	\kappa_i(e_i, \overline{e}_i) \coloneqq -k_{i} (e_i-\overline{e}_i) = - k_{i} z_i,
	\end{equation}	
	where the control gains $k_{i} > 0$ are chosen as:
	\begin{align} 
	k_{i} \coloneqq \underline{k}_i + \frac{1}{\underline{J}_i} \left(L_i+\frac{1}{4 \rho_i}\right) > 0, \label{eq:control_gains2}
	\end{align}
	where $\underline{k}_i, \rho_i > 0$, are parameters to be appropriately tuned, render the sets:
	\begin{equation} \label{eq:Omega_set}
	\mathcal{Z}_i \coloneqq \left\{z_{i}(t) \in \mathbb{R}^n : \|z_{i}(t)\| \le  \widetilde{z}_i \coloneqq \frac{\sqrt{\rho_i} \widetilde{w}_i}{\sqrt{\underline{k}_i \underline{J}_i}}, \forall t \in \mathbb{R}_{\ge 0} \right\},
	\end{equation}
	RCI sets for the error dynamical systems \eqref{eq:z_dynamics}, according to Definition \ref{def:RPI_set}.
\end{lemma}
\begin{proof}
	The proof is given in Appendix \ref{appendix:proof_lemma_proposed_RPI_set}.
\end{proof}

The aforementioned result is crucial in robust NMPC frameworks. As it will be shown hereafter, the nominal control actions $\overline{u}_i(t) \in \mathcal{U}_i$ will be the solution of a repetitive DFHOCP solved at every sample time from each agent $i \in \mathcal{V}$. The nominal trajectories $\overline{e}_i(t)$ computed by the DFHOCP define a hyper-tube in $\mathcal{D}$ centered along them with radius $\widetilde{z}_i$, for every $t \in \mathbb{R}_{>0}$ and $i \in \mathcal{V}$ (see Figure \ref{fig: example_01}). By using \eqref{eq:control_law_u_kappa}, \eqref{eq:kappa_law}, the \emph{closed-loop system} of each agent is written as:
\begin{equation} \label{eq:closed_loop_system}
\dot{e}_i(t) = f_i\Big(e_i(t), \overline{u}_i(t) - k_i \left[e_i(t)-\overline{e}_i(t)\right] \Big) + w_i(t), i \in \mathcal{V}.
\end{equation} 

\begin{remark}
	The volume of the hyper-tubes depends on the bound of the disturbances $\widetilde{w}_i$, the Lipschitz constant $L_{i}$ of functions $f_i$ as well as the constants $\underline{J}_i$. Moreover, by tuning the parameters $\underline{k}_i$, $\rho_i$ appropriately, we can adjust the volume of the tube.
\end{remark}

\subsection{Decentralized Online Control Design of $\overline{u}_i(t)$} \label{sec:online_control_design}

Consider a sequence of sampling times $\{t_k\}$, $k \in \mathbb{N}$, with a constant sampling time $\delta$, $0 < \delta < T$, where $T$ is the prediction horizon, such that $t_{k+1} \coloneqq t_k + \delta$, $\forall k \in \mathbb{N}$. In sampled data NMPC, a DFHOCP is solved at discrete sampling time instants $t_k$, by each agent $i \in \mathcal{V}$, based on the current state error measurement $\overline{e}_i(t_k)$. The solution is an optimal control signal $\overline{u}_{i}^{\star}(s)$, computed over $s \in [t_k, t_k+T_p]$. The open-loop input signal applied in between the sampling instants is given by the solution of the following DFHOCP:
\begin{subequations}
	\begin{align}
	&\hspace{-1mm}\min\limits_{\overline{u}_i(\cdot)} J_i(\overline{e}_i(t_k), \overline{u}_i(\cdot)) \notag \\ 
	&\hspace{-1mm}= \min\limits_{\overline{u}_i(\cdot)} \left\{  V_i(\overline{e}_i(t_k+T)) + \int_{t_k}^{t_k+T} \Big[ F_i(\overline{e}_i(s), \overline{u}_i(s)) \Big] ds \right\}  \label{eq:mpc_cost_function} \\
	&\hspace{-1mm}\text{subject to:} \notag \\
	&\hspace{1mm} \dot{\overline{e}}_i(s) = f_i(\overline{e}_i(s)+x_{\scriptscriptstyle i, \text{des}}, \overline{u}_i(s)), \overline{e}_i(t_k) = e_i(t_k), \label{eq:diff_mpc} \\
	&\hspace{1mm} \overline{e}_i(s) \in \overline{\mathcal{E}}_i, \overline{u}_i(s) \in \overline{\mathcal{U}}_i, s \in [t_k,t_k+T], \label{eq:mpc_constrained_set} \\
	&\hspace{1mm} \overline{e}_i(t_k+T)\in \mathcal{F}_i. \label{eq:mpc_terminal_set}
	\end{align}
\end{subequations}
At a generic sampling time $t_k$, agent $i \in \mathcal{V}$ solves the aforementioned DFHOCP. This means that $\overline{e}_i(\cdot)$ is the solution to \eqref{eq:diff_mpc} driven by the control input $\overline{u}_i(\cdot) : [t_k, t_k + T] \to \overline{\mathcal{U}}_i$ with initial condition $e_i(t_k)$. The functions $F_i: \mathbb{R}^n \times \mathbb{R}^n \to \mathbb{R}_{\geq 0}$, $V_i: \mathbb{R}^n \to \mathbb{R}_{\geq 0}$ stand for the \emph{running cost} and the \emph{terminal penalty cost}, respectively, and they are defined by: $F_i(e_i, u_i) \coloneqq e_i^{\top} Q_i e_i + u_i^{\top} R_i u_i$ and $V_i(e_i) = e_i^{\top} P_i e_i$, respectively; $R_i \in \mathbb{R}^{n \times n}$ and $Q_i$, $P_i \in \mathbb{R}^{n \times n}$ are positive definite gain matrices to be appropriately tuned; $Q_i \in \mathbb{R}^{n \times n}$ is a positive semi-definite gain matrix to be appropriately tuned. The sets $\mathcal{F}_i$ are the terminal sets that are used to force the stability of the multi-agent system and will be explained later. 

We will explain hereafter the form of the sets $\overline{\mathcal{E}}_i$ and $\overline{\mathcal{U}}_i$. In order to guarantee that while each agent $i \in \mathcal{V}$ is solving the DFHOCP \eqref{eq:mpc_cost_function}-\eqref{eq:mpc_terminal_set} for its nominal system \eqref{eq:error_system_nominal}, the real system trajectories $e_i(t)$, which are the solution of \eqref{eq:error_system_perturbed} as well as the control inputs $u_i(t)$ satisfy the state and input constraints $\mathcal{E}_i$ and $\mathcal{U}_i$, respectively, the sets $\mathcal{E}_i$ and $\mathcal{U}_i$ needs to be properly modified. Due to the fact that $\mathcal{Z}_i$ are RCI sets of the error dynamics \eqref{eq:z_dynamics}, under the control law \eqref{eq:control_law_u_kappa}, \eqref{eq:kappa_law}, it holds that:
\begin{align}
e_i(s)-\overline{e}_i(s) \in \mathcal{Z}_i & \Rightarrow e_i(s) \in \mathcal{Z}_i \oplus \left\{ \overline{e}_i(s) \right\}, \forall s \in [t_k, t_k+T], i \in \mathcal{V}. \label{eq:prop_sets_item1}
\end{align}
Combining the latter with the fact that $e_i(s)$ needs to satisfy the state constraints $\mathcal{E}_i$, $\forall s \in [t_k, t_k+T]$, $i \in \mathcal{V}$, the state constraints set of each agent are modified as: $$\overline{\mathcal{E}}_i \coloneqq \mathcal{E}_i \ominus \mathcal{Z}_i.$$ Moreover, by using \eqref{eq:prop_sets_item1} as well as \eqref{eq:control_law_u_kappa} and \eqref{eq:kappa_law}, we have:
\begin{align}
u_i(s) - \overline{u}_i(s) = (-k_i) \left(e_i(s) - \overline{e}_i(s)\right) & \in (-k_i) \circ \mathcal{Z}_i, \forall s \in [t_k, t_k + T] \notag \\ 
\Rightarrow u_i(s) & \in \left[(-k_i) \circ \mathcal{Z}_i \right] \oplus \overline{u}_i(s), \forall s \in [t_k, t_k+T], i \in \mathcal{V}. \label{eq:prop_sets_item2}
\end{align}
Combining the latter with the fact that $u_i(s)$ needs to satisfy the input constraints $\mathcal{U}_i$, $\forall s \in [t_k, t_k+T]$, $i \in \mathcal{V}$, the input constraints set of each agent is modified as: $$\overline{\mathcal{U}}_i \coloneqq \mathcal{U}_i \ominus \left[ (-k_i) \circ \mathcal{Z}_i \right].$$ Intuitively, the sets $\mathcal{E}_i$ and $\mathcal{U}_i$ are tightened in order to guarantee that while the nominal trajectory $\overline{e}_i(t)$ and the nominal input $\overline{u}_i$ are calculated, the corresponding real trajectory $e_i(t)$ and input $u_i(t)$ satisfy the state and input constraints $\mathcal{U}_i$, $\mathcal{E}_i$, respectively. The advantage of the tube-based frameworks compared to other robust NMPC approaches, is that the constraint tightening is performed off-line, it does not depend of the length of the horizon, but it depends only on the radius of the hyper-tubes $\widetilde{z}_i$ and the control gains $k_i$.

Due to the fact that the connectivity between initially connected agents, i.e, $j \in \mathcal{N}_i$, $\forall i \in \mathcal{V}$, needs to be preserved and the agents have communication capabilities within the sensing range $d_i$ as given in \eqref{eq:d_i}, we adopt here the decentralized procedure depicted in Algorithm $1$ and explained hereafter. Assume that each agent knows its labeling number of the set $\mathcal{V}$. After each sampling time $t_k$, $\forall k \ge 0$ that agent $i$ solves its own DFHOCP and obtains the estimated open-loop trajectory $\overline{e}_i(s)$, $s \in [t_k, t_k+T]$, it transmits it to its neighboring agents $j \in \mathcal{N}_i$. Then, agents' $j \in \mathcal{N}_i$ hard constraints $\overline{\mathcal{E}}_{j}$ are updated by incorporating the predicted trajectory of agent $i$, i.e., $\overline{e}_i(s)$, $s \in [t_k, t_k+T]$. Among all agents $j \in \mathcal{N}_i$, the one with higher priority, i.e., smaller labeling number in the set $\mathcal{V}$, solves its own DFHOCP (for example, agent $2$ has higher priority than agents $3$, $4$, $\dots$). This \emph{sequential procedure} is continued until all agents $i \in \mathcal{V}$ solve their own DFHOCP, and then the sampling time is updated. We will show thereafter that by adopting the aforementioned sequential communication procedure, and given that at $t = 0$ the DFHOCP \eqref{eq:mpc_cost_function} - \eqref{eq:mpc_terminal_set} of all agents are feasible, the agents are navigated to their desired configurations $x_{\scriptscriptstyle i, \rm des}$, while all distance and input constraints imposed by Problem \ref{problem} are satisfied. Note that exchanging open-loop predicted trajectories between neighboring agents is an approach adopted earlier in decentralized multi-agent MPC frameworks (see e.g., \cite{muller2012cooperative}).

\begin{algorithm}[t!]
	\caption{Information exchange procedure within a horizon $T$}
	\begin{algorithmic}[1]
		\STATE $s \in [t_k, t_k+T]$;
		\WHILE {$\mathcal{V} \neq \emptyset$}
		\FOR {$i \in \mathcal{V}$}
		\STATE Solve DFHOCP \eqref{eq:mpc_cost_function} - \eqref{eq:mpc_terminal_set} for agent $i$;
		\STATE Transmit $\overline{e}_i(s)$, $s \in [t_k, t_k+T]$ to all neighbors $j \in \mathcal{N}_i$;
		\FOR {$j \in \mathcal{N}_i$}
		\STATE Update $\overline{\mathcal{E}}_j$; \hspace{17mm} \COMMENT{Agent $j$ has acess to open loop predictions of agent $i$}
		\STATE Solve DFHOCP \eqref{eq:mpc_cost_function} - \eqref{eq:mpc_terminal_set} for agent $j$;
		\IF {$\mathcal{N}_j \backslash \{i\} = \emptyset$} 
		\STATE $\mathcal{V} \leftarrow \mathcal{V} \backslash \{j\}$; \hspace{49mm} \COMMENT{Agent $j$ has no other neighbors}
		\ENDIF
		\ENDFOR 
		\ENDFOR
		\STATE $\mathcal{V} \leftarrow \mathcal{V} \backslash \{i\}$;
		\ENDWHILE
		\STATE $t_k \leftarrow t_k + \delta$; \hspace{41mm} \COMMENT{Update sampling time and run the procedure}
		\STATE \textbf{Go to} ``line $1$";   
	\end{algorithmic} 
\end{algorithm}

The nominal controller $\overline{u}_i(\cdot)$ of each agent $i \in \mathcal{V}$ is calculated online. The state feedback control law $\kappa_i(e_i, \overline{e}_i)$ defined in \eqref{eq:kappa_law}, is obtained offline, and keeps the trajectories of the error system \eqref{eq:z_dynamics} in a hyper-tube centered along the nominal trajectory $\overline{e}_i(s)$ with radius $\widetilde{z}_i$ as given in \eqref{eq:Omega_set}. We will show hereafter that the proposed control law \eqref{eq:control_law_u_kappa}, where $\overline{u}_i(\cdot)$ is the solution of the DFHOCP \eqref{eq:mpc_cost_function} - \eqref{eq:mpc_terminal_set} for the nominal system \eqref{eq:error_system_nominal}, navigates each agent $i \in \mathcal{V}$ with dynamics as in \eqref{eq:sys} to its desired configuration $x_{\scriptscriptstyle i, \rm des}$, for all $w_i(t) \in \mathcal{W}_i$, $t \in \mathbb{R}_{\ge 0}$. 

The solution to DFHOCP \eqref{eq:mpc_cost_function} - \eqref{eq:mpc_terminal_set} at time $t_k$ provides an optimal control input, denoted by
$\overline{u}_i^{\star}(s;\ e_i(t_k))$, $s \in [t_k, t_k + T]$. This control input is then applied to the system until the next sampling instant:
\begin{align}
\overline{u}_i(s; \overline{e}_i(t_k)) \coloneqq \overline{u}_i^{\star}\big(s; \ e_i(t_k)\big),\  s \in [t_k, t_k+\delta).
\label{eq:position_based_optimal_u_2}
\end{align}
At time $t_{k+1} = t_k + \delta$ a new DFHOCP is solved in the same manner, leading to a receding horizon approach.  Algorithm $2$ (from \cite{yu_2013_tube}) depicts the procedure of how the proposed control law is calculated and applied to the real system. The solution of \eqref{eq:error_system_perturbed} at time $s$, $s \in [t_k, t_k+T]$, starting at time $t_k$, from an initial condition $\overline{e}_i(t_k) = e_i(t_k)$, by application of the control input $u_i: [t_k, s] \to \overline{\mathcal{U}}_i$ as in \eqref{eq:overal_control_input_system}, is denoted by $e_i\big(s;\ u_i(\cdot), \overline{e}_i(t_k)\big)$, $s \in [t_k, t_k+T]$. The \textit{predicted} state of the system \eqref{eq:diff_mpc} at time $s$ based on the measurement of the state at time $t_k$, $\overline{e}_i(t_k)$, by application of the control input $\overline{u}_i\big(t;\ \overline{e}_i(t_k)\big)$ as in \eqref{eq:position_based_optimal_u_2}, is denoted by $\overline{e}_i\big(s;\ \overline{u}_i(\cdot), \overline{e}_i(t_k)\big)$, $s \in [t_k, t_k+T]$. The overall applied control input for the actual system \eqref{eq:sys} during the sampling interval consequently is:
\begin{equation} \label{eq:overal_control_input_system}
u_i(s) = \overline{u}_i(s; \overline{e}_i(t_k)) + \kappa_i(e_i(s), \overline{e}_i\big(s;\ \overline{u}_i(\cdot), \overline{e}_i(t_k))), s \in [t_k, t_k + \delta),
\end{equation}
where $\overline{u}_i(s; \overline{e}_i(t_k))$ is the optimal input given by \eqref{eq:position_based_optimal_u_2}.

\begin{definition}  \label{definition:admissible_input_with_disturbance}
	A control input $\overline{u}_i : [t_k, t_k + T] \to \mathbb{R}^n$ for a state $e_i(t_k)$ of agent $i \in \mathcal{V}$ is called \textit{admissible} for the DFHOCP \eqref{eq:mpc_cost_function}-\eqref{eq:mpc_terminal_set} if the following hold:
	\begin{enumerate}
		\item $\overline{u}_i(\cdot)$ is piecewise continuous;
		\item $\overline{u}_i(s) \in \overline{\mathcal{U}}_i,\ \forall s \in [t_k, t_k + T]$;
		\item $\overline{e}_i\big(t_k + s;\ \overline{u}_i(\cdot), \overline{e}_i(t_k)\big) \in \overline{\mathcal{E}}_i,\ \forall s \in [0, T]$; and
		\item $\overline{e}_i\big(t_k + T;\ \overline{u}_i(\cdot), \overline{e}_i(t_k)\big) \in \mathcal{F}_i$.
	\end{enumerate}
\end{definition}

Define the terminal set of each agent $i \in \mathcal{V}$ as: $\mathcal{F}_i \coloneqq \{\bar{e}_i \in \overline{\mathcal{E}}_i: V_i(\overline{e}_i) \le \eta_i\} \subseteq \overline{\mathcal{E}}_i$, with $\eta_i > 0$. Then, according to Assumption \ref{ass:stabilized_linear_assumption}, (we refer the reader to \cite{frank_1998_quasi_infinite} for more details), there exists a \emph{local admissible feedback law} $u_{i, \text{loc}}(\overline{x})$ which guarantees that:
\begin{enumerate}
	\item $u_{i, \text{loc}}(\overline{e}_i) \in \overline{\mathcal{U}}_i$, for every $\overline{e}_i \in \mathcal{F}_i$;
	\item $\displaystyle \frac{\partial V_i(\overline{e}_i)}{\partial \overline{e}_i} f_i(\overline{e}_i, u_{i, \text{loc}}(\overline{e}_i)) + F_i(\overline{e}_i, u_{i, \rm loc}(\overline{e}_i)) \le 0$, for every $\overline{e}_i \in \mathcal{F}_i$.
\end{enumerate}

%Note that the latter implies that the sets $\mathcal{F}_i$ are invariant for the nominal system \eqref{eq:error_system_nominal}. 
The methodology under which the constant $\eta_i > 0$ as well as the local controller $u_{i, \text{loc}}$ are chosen can be found in \cite{frank_1998_quasi_infinite}. Under these considerations, we can now state the theorem that guarantees the stability of the system \eqref{eq:sys}, under the proposed control law \eqref{eq:control_law_u_kappa}, \eqref{eq:kappa_law} for all initial conditions $e_i(0) \in \mathcal{E}_i$, $i \in \mathcal{V}$.

\begin{algorithm}[t!]
	\caption{Implementation of feedback control laws $u_i(t)$, $i \in \mathcal{V}$}
	\begin{algorithmic}
		\STATE \textbf{Step $\mathbf{0}$:} At time $t_0 \coloneqq 0$, set $e_i(0) = \overline{e}_i(0)$ where $e_i(0)$ is the current state.
		\STATE \textbf{Step $\mathbf{1}$:} At time $t_k$ and current state $(e_i(t_k), \overline{e}_i(t_k))$, solve DFHOCP \eqref{eq:mpc_cost_function}-\eqref{eq:mpc_terminal_set} to obtain the nominal control action $\overline{u}_i(t_k)$ and the actual control action $u_i(t_k) = \overline{u}_i(t_k)+\kappa_i(e_i(t_k), \overline{e}_i(t_k))$.
		\STATE \textbf{Step $\mathbf{2}$:} Apply the control $u_i(t_k)$ to the system \eqref{eq:error_system_perturbed}, during sampling interval $[t_k, t_{k+1})$, where $t_{k+1} = t_{k}+\delta$.
		\STATE \textbf{Step $\mathbf{3}$:} Measure the state $e_i(t_{k+1})$ at the next time instant $t_{k+1}$ of the system \eqref{eq:error_system_perturbed} and compute the successor state $\overline{e}_i(t_{k+1})$ of the nominal system \eqref{eq:error_system_nominal} under the nominal control action $\overline{u}_i(t_k)$.
		\STATE \textbf{Step $\mathbf{4}$:} Set $(e_i(t_k), \overline{e}_i(t_k)) \leftarrow (e_i(t_{k+1}), \overline{e}_i(t_{k+1}))$, $t_{k} \leftarrow t_{k+1}$; 
		
		\textbf{Go to} \textbf{Step $\mathbf{1}$}.
	\end{algorithmic} 
\end{algorithm}

\begin{theorem} \label{theorem}
	\label{theorem:with_disturbances}
	Suppose that Assumptions \ref{ass:f_assumption}-\ref{ass:lower_bound_deriv} hold. Suppose also that at time $t = 0$ the DFHOCP \eqref{eq:mpc_cost_function}-\eqref{eq:mpc_terminal_set} sequentially solved by all the agents $i \in \mathcal{V}$, is feasible. Then, the proposed decentralized feedback control law \eqref{eq:control_law_u_kappa}, \eqref{eq:kappa_law}, renders the closed-loop system \eqref{eq:closed_loop_system} of each agent $i \in \mathcal{V}$ Input to State Stable with respect to $w_i(t) \in \mathcal{W}_i$, for every initial conditions $e_i(0) \in \mathcal{E}_i$.
\end{theorem}

\begin{proof}
	The proof of the theorem consists of two parts: 
	
	\noindent \textbf{Recursive Feasibility :} It will be shown that recursive feasibility is established, and it implies subsequent feasibility. The feasibility proof can be found in Appendix \ref{app:feasibility_analysis};
	
	\noindent \textbf{Convergence Analysis :} The convergence analysis is given in Appendix \ref{app:convergence_analysis}. 
\end{proof}

\begin{remark}
	Assumption \ref{ass:f_assumption} - \ref{ass:stabilized_linear_assumption} as well as communication capabilities among the agents are standard assumptions in order for the nominal stability of decentralized NMPC schemes to be guaranteed. We refer the reader to \cite{frank_1998_quasi_infinite, muller2012cooperative} for more details.
\end{remark}

\section{Simulation Results} \label{sec:simulation_results}

\noindent For a simulation scenario, consider $N = 3$ agents $\mathcal{V} = \{1,2,3\}$ with uncertain non-affine nonlinear dynamics given as follows:
\begin{align*}
\dot{x}_{i,1} & = \frac{0.1-0.1e^{-x_{i,2}}}{1+e^{-x_{i,2}}} + 0.25 x_{i,1}^2 + 2 u_{i,1} +0.3 \cos(t), \notag \\
\dot{x}_{i,2} & = 0.25 x_{i,1}^2 + u_{i,2} + 0.1 u_{i,2}^{3} + 0.2 \sin(2t), \notag
\end{align*} 
where: 
\begin{align*}
x_i & = [x_{i,1}, x_{i,2}]^\top \in \mathbb{R}^2, \notag \\ 
u_i & =[u_{i,1}, u_{i,2}]^\top \in \mathbb{R}^2, \notag \\
f_i(x_i, u_i) & = 
\begin{bmatrix}
f_{i,1}(x_i,u_i) \\
f_{i,2}(x_i,u_i)
\end{bmatrix}
=
\begin{bmatrix}
\frac{0.1-0.1e^{-x_{i,2}}}{1+e^{-x_{i,2}}} + 0.25 x_{i,1}^2 + 2 u_{i,1} \\
0.25 x_{i,1}^2 + u_{i,2} + 0.1 u_{i,2}^{3}
\end{bmatrix}, \notag \\
w_i(t) & = [0.3 \cos(t), 0.2 \sin(2t)]^\top,
\end{align*}
with $\|w_i(t)\| \le 0.3 = \widetilde{w}_i$, $\forall t \in \mathbb{R}_{\ge 0}$. From \eqref{eq:function_J_i} we get:
\begin{align*}
J_i(x_i, u_i) & = \sum_{k = 1}^{2} \sum_{j=1}^{2} \ell_n(k) \ell_n(j)^\top \frac{\partial f_{i,k}(x_i, u_i)}{\partial u_j} \notag \\
& =
\begin{bmatrix}
2 & 0 \\
0 & 1 + 0.3 u_{i,2}^2
\end{bmatrix},
\end{align*}
with $\lambda_{\min}\left(\frac{J_i+J_i^\top}{2}\right) \ge \underline{J}_i = 1$. The agents are operating in a workspace $\mathcal{D} = \{x_i \in \mathbb{R}^2 : -5 \le x_{i,1}, x_{i,2} \le 5\}$ with $L_i = 2.5$. The radius and the sensing range of all agents are set to $r_i = 1$ and $d_i = 5$, respectively. We set $\varepsilon = 0.01$, where $\varepsilon$ is the parameter of the constraints set \eqref{eq:constraints_X}. The sensing radii result to the following neighboring sets: $\mathcal{N}_1 = \{2\} = \mathcal{N}_3$ and $\mathcal{N}_2 = \{1,3\}$. The agents' initial positions are $x_{1}(0) = [-3.0, 2.9]^\top$, $x_{2}(0) = [-2.5, -0.2]^\top$ and $x_{3}(0) = [-2.9, -4]^\top$. Their corresponding desired configurations are $x_{1, \text{des}} = [0.1206, 1.1155]^\top$, $x_{2, \text{des}} = [2.0, 0.0]^\top$ and $x_{3, \text{des}} = [0.9, -2.8]^\top$. According to Definition \ref{definition:feasibl_configurations}, the above configurations are feasible since it holds that: $\|x_{\scriptscriptstyle i, \rm des} - x_{\scriptscriptstyle j, \rm des}\| < d_i$, $\forall i \in \mathcal{V}$, $j \in \mathcal{N}_i$ .The sampling time and the total execution time are $\delta = 0.1$ and $10 \sec$, respectively. The control gains are chosen as $\rho_i = \underline{k}_i = 1$ and $k_i = 3.75$, which result to a tube of radius $\widetilde{z}_i = 0.3$. The matrices $Q_i$, $R_i$ and $P_i$ are set to $0.5 I_2$. The initial error constraints of each agent are given as:
\begin{align}
\mathcal{E}_1 & = \{e_1 \in \mathbb{R}^2 : -5.1206 \le e_{1,1} \le 4.8794, -6.1155 \le e_{1,2} \le 3.8845 \}, \notag \\
\mathcal{E}_2 & = \{e_2 \in \mathbb{R}^2 : -7.0 \le e_{2,1} \le 3.0, -5.0 \le e_{2,2} \le 5.0 \}, \notag \\
\mathcal{E}_3 & = \{e_3 \in \mathbb{R}^2 : -5.9 \le e_{3,1} \le 4.1, -2.2 \le e_{3,2} \le 2.2 \}, \notag
\end{align}
and the corresponding modified error constrains which are used for the solution of the online NMPC as:
\begin{align}
\overline{\mathcal{E}}_1 & = \{e_1 \in \mathbb{R}^2 : -4.8206 \le e_{1,1} \le 4.5794, -5.8155 \le e_{1,2} \le 3.5845 \}, \notag \\
\overline{\mathcal{E}}_2 & = \{e_2 \in \mathbb{R}^2 : -6.7 \le e_{2,1} \le 2.7, -4.7 \le e_{2,2} \le 4.7 \}, \notag \\
\overline{\mathcal{E}}_3 & = \{e_3 \in \mathbb{R}^2 : -5.6 \le e_{3,1} \le 3.8, -1.9 \le e_{3,2} \le 1.9 \}. \notag
\end{align}
The input constraints of each agent are set to: $$\mathcal{U}_i = \{u_i \in \mathbb{R}^{2} : -2.125 \le u_{i,1}, u_{i,2} \le 2.125\}, i \in \mathcal{V}.$$ The corresponding modified input constraint sets for the online NMPC are given as: $$\overline{\mathcal{U}}_i = \{\overline{u}_i \in \mathbb{R}^2: -1 \le \overline{u}_{i,1}, \overline{u}_{i,2} \le 1\}, i \in \mathcal{V}.$$

Fig. \ref{fig:tube1}, \ref{fig:tube2} and \ref{fig:tube3} depict the evolution of the real and the nominal trajectories as well as the tubes of agents $1$, $2$ and $3$ respectively. The tubes are centered along the nominal trajectory of each agent and the real trajectory always remain within the tubes for all times. Fig. \ref{fig:error1}, \ref{fig:error2} and \ref{fig:error3} represent the evolution of the error signals $e_1(t)$, $e_2(t)$ and $e_3(t)$, respectively. The evolution of the trajectories of all agents in the workspace is depicted in Fig. \ref{fig:workspace}. The distance between the neighboring agents $1-2$ and $2-3$ is represented in Fig. \ref{fig:distance}. Finally, the control effort of each agent is shown in Fig. \ref{fig:inputs}. 

It can be observed that all agents reach their desired configurations by satisfying all the constraints imposed by Problem $1$. The simulation was performed in MATLAB R2015a Environment utilizing the NMPC optimization toolbox provided in \cite{grune_2011_nonlinear_mpc}. The simulation takes $82 \sec$ on a desktop with $8$ cores, $3.60$ GHz CPU and $16$GB of RAM.

\begin{remark}
	It should be noted that in this paper we consider heterogeneous agents, i.e., functions $f_i$, $i \in \mathcal{V}$ in \eqref{eq:sys} may be different for each agent. In the aforementioned simulation example, for convenience and simplified calculations, we considered homogeneous agents, i.e., the functions $f_i$ are the same for all agents.
\end{remark}

\begin{figure}[t!]
	\centering
	\includegraphics[scale = 0.6]{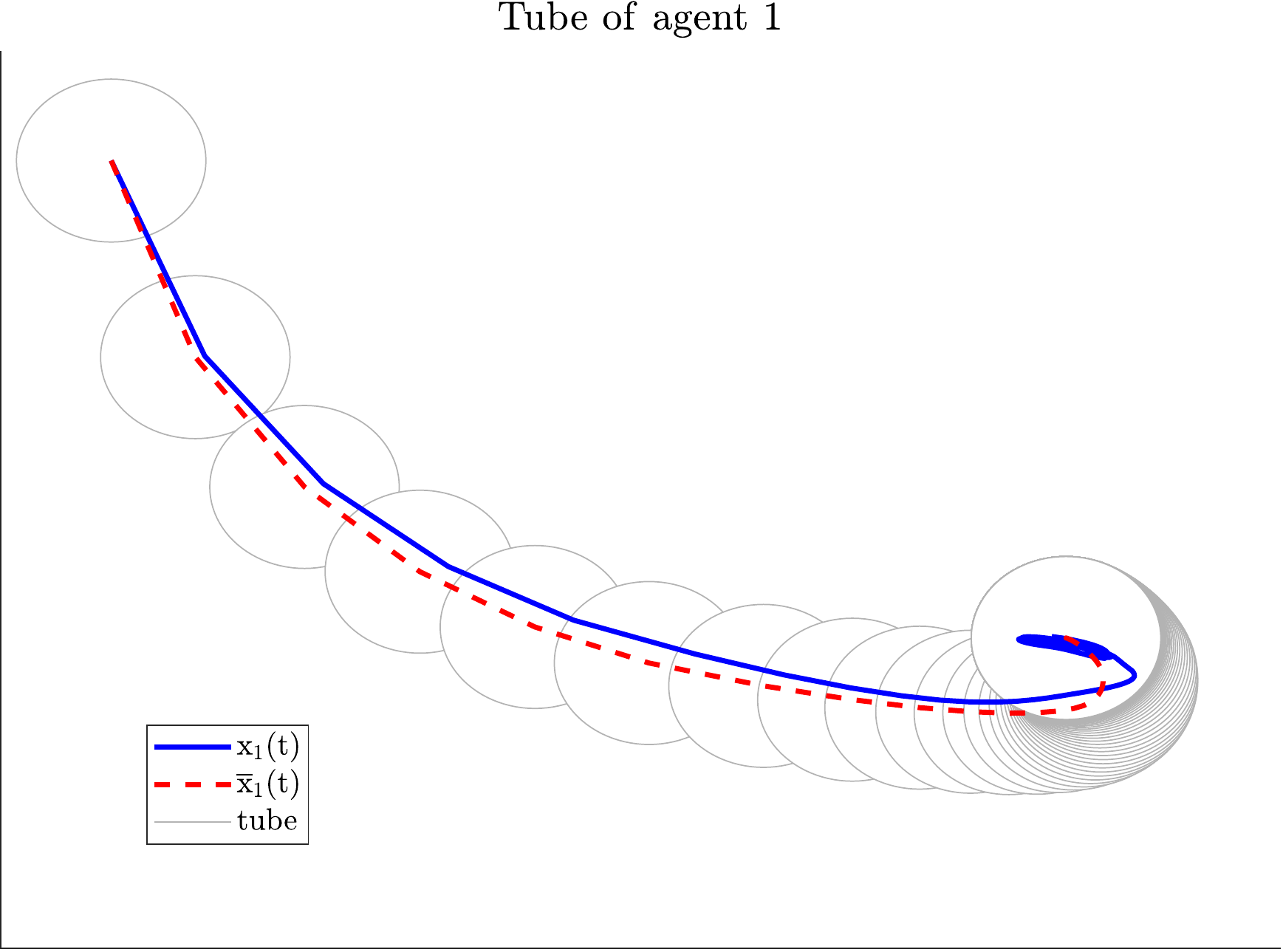}
	\caption{The evolution of the trajectory of agent $1$ in the workspace $\mathcal{D}$ over the time interval $[0,10] \sec$. The solid and the dashed lines represent the real and the nominal trajectory, respectively. The gray circles represent the tube.}\label{fig:tube1}
\end{figure}

\begin{figure}[t!]
	\centering
	\includegraphics[scale = 0.6]{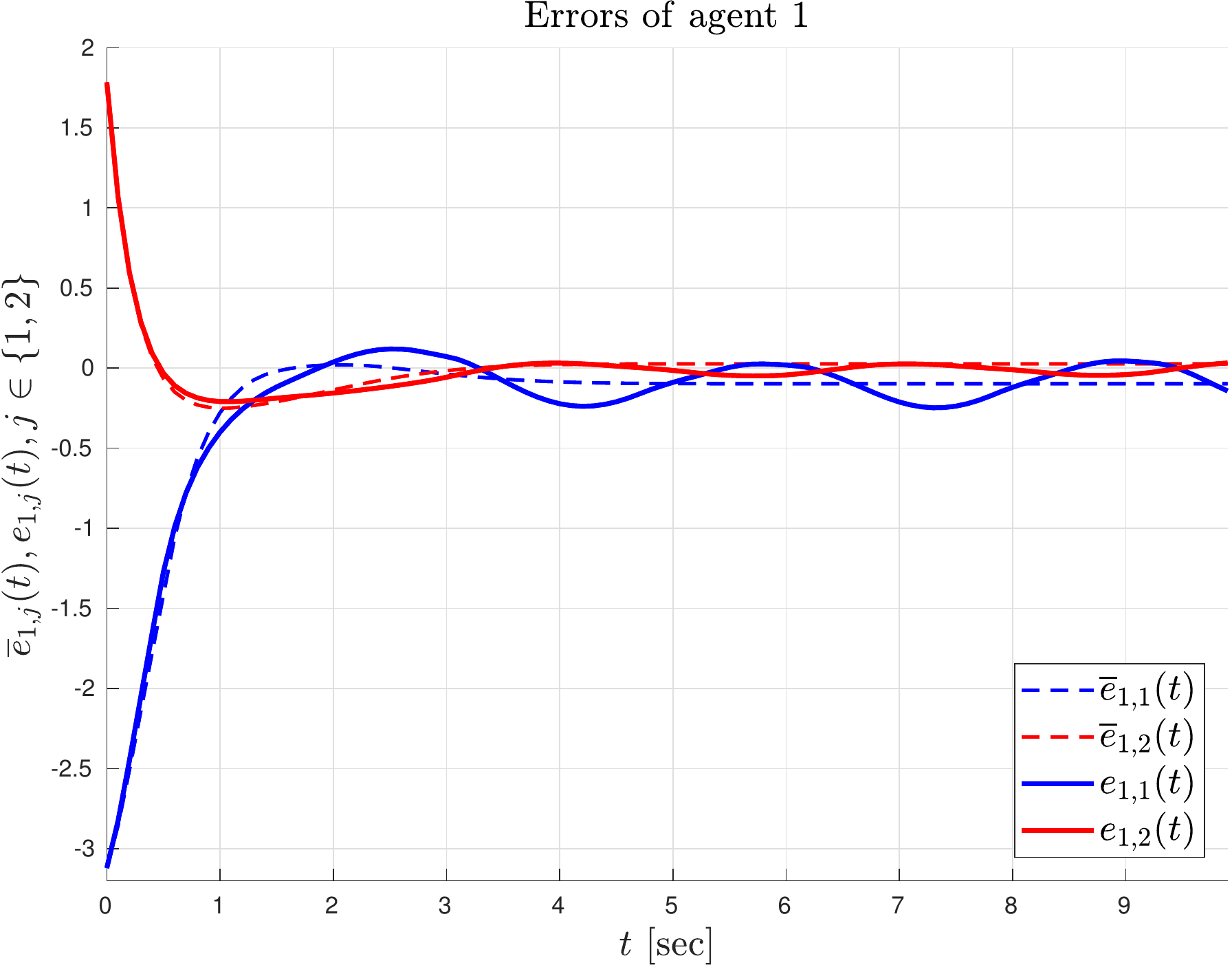}
	\caption{The evolution of the real error signals $e_1(t)$ and the nominal error signals $\overline{e}_1(t)$ over the time interval $[0,10] \sec$}\label{fig:error1}
\end{figure}

\begin{figure}[t!]
	\centering
	\includegraphics[scale = 0.6]{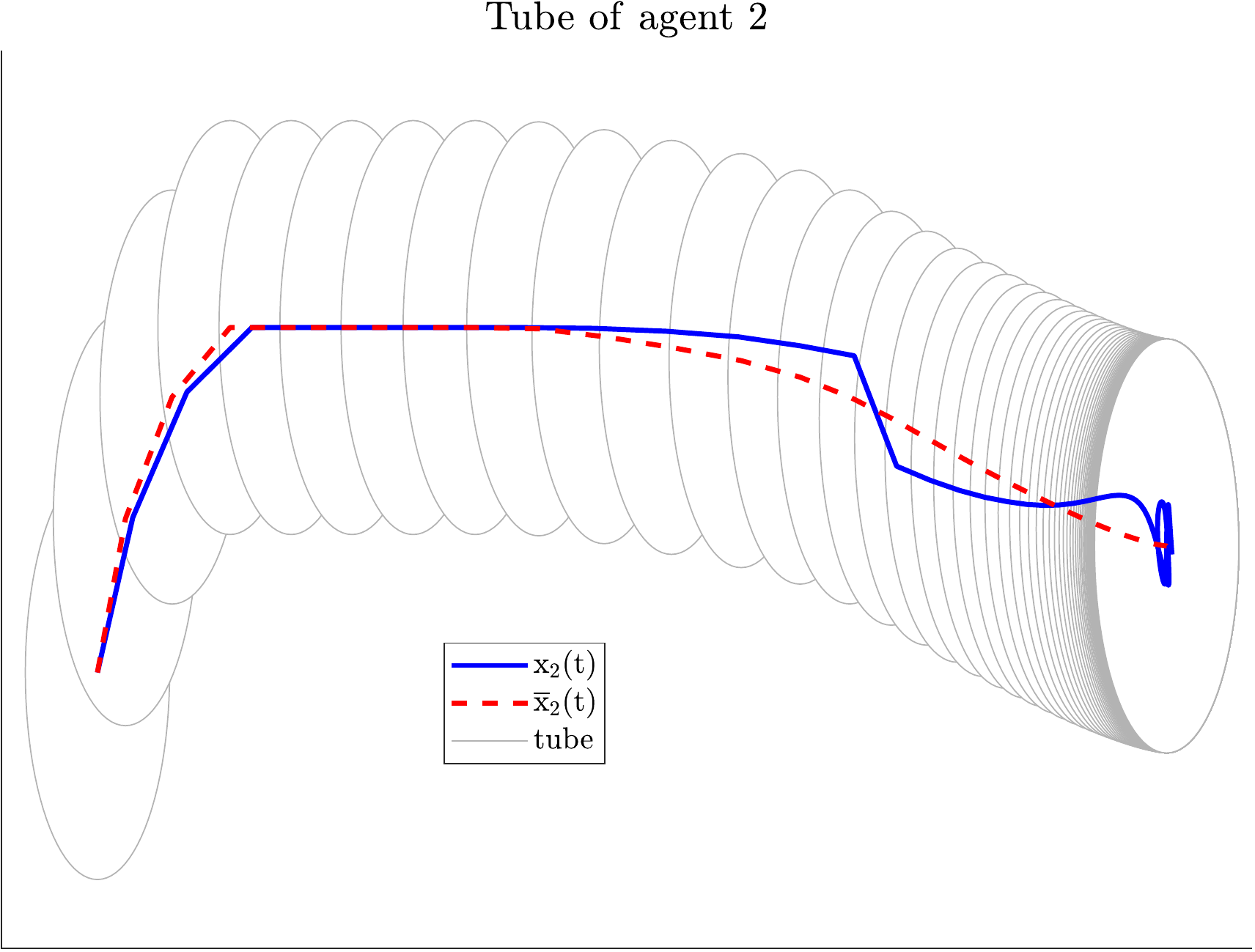}
	\caption{The evolution of the trajectory of agent $2$ in the workspace $\mathcal{D}$ over the time interval $[0,10] \sec$. The solid and the dashed lines represent the real and the nominal trajectory, respectively. The gray circles represent the tube.}\label{fig:tube2}
\end{figure}

\begin{figure}[t!]
	\centering
	\includegraphics[scale = 0.6]{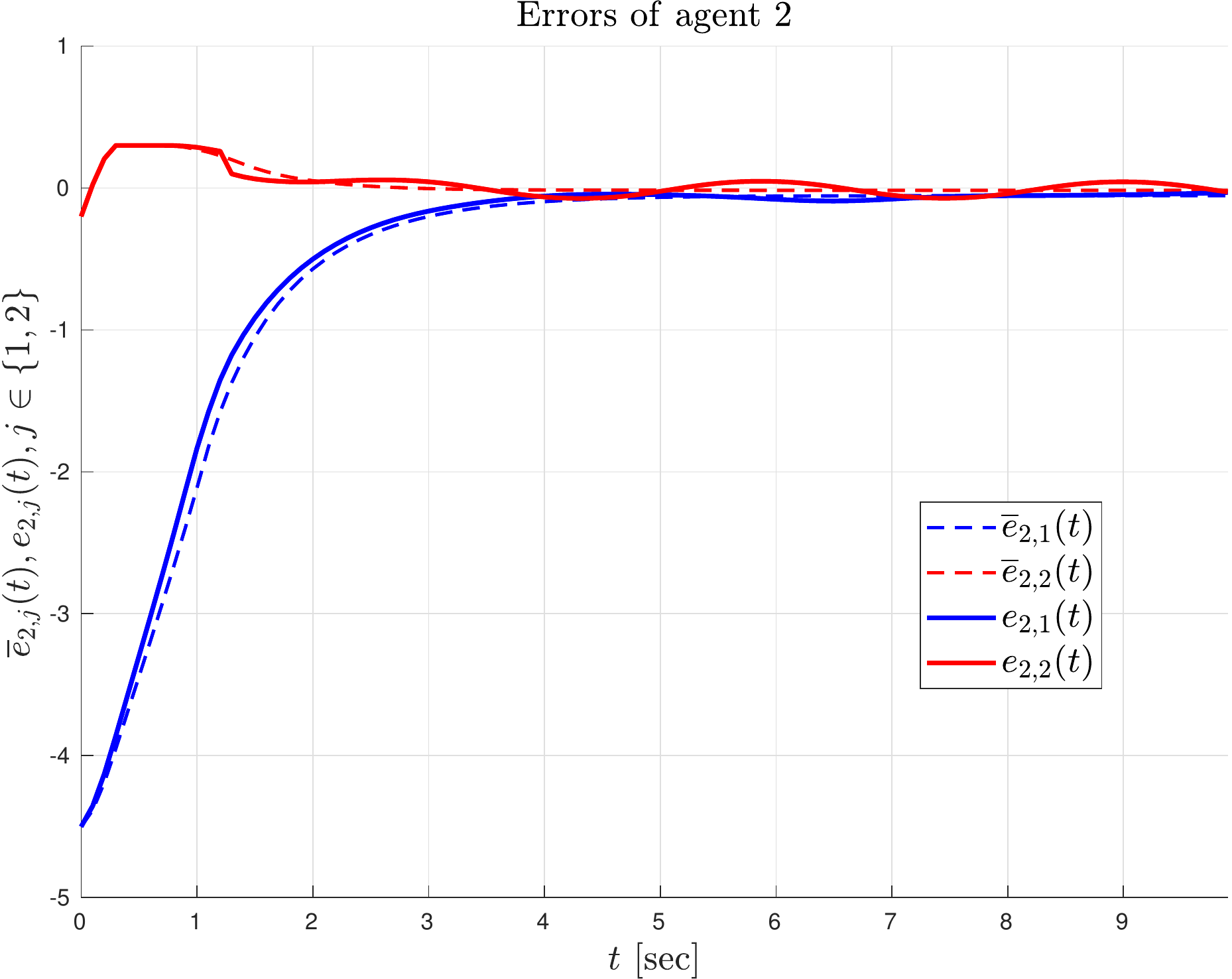}
	\caption{The evolution of the real error signals $e_2(t)$ and the nominal error signals $\overline{e}_2(t)$ over the time interval $[0,10] \sec$.}\label{fig:error2}
\end{figure}

\begin{figure}[t!]
	\centering
	\includegraphics[scale = 0.55]{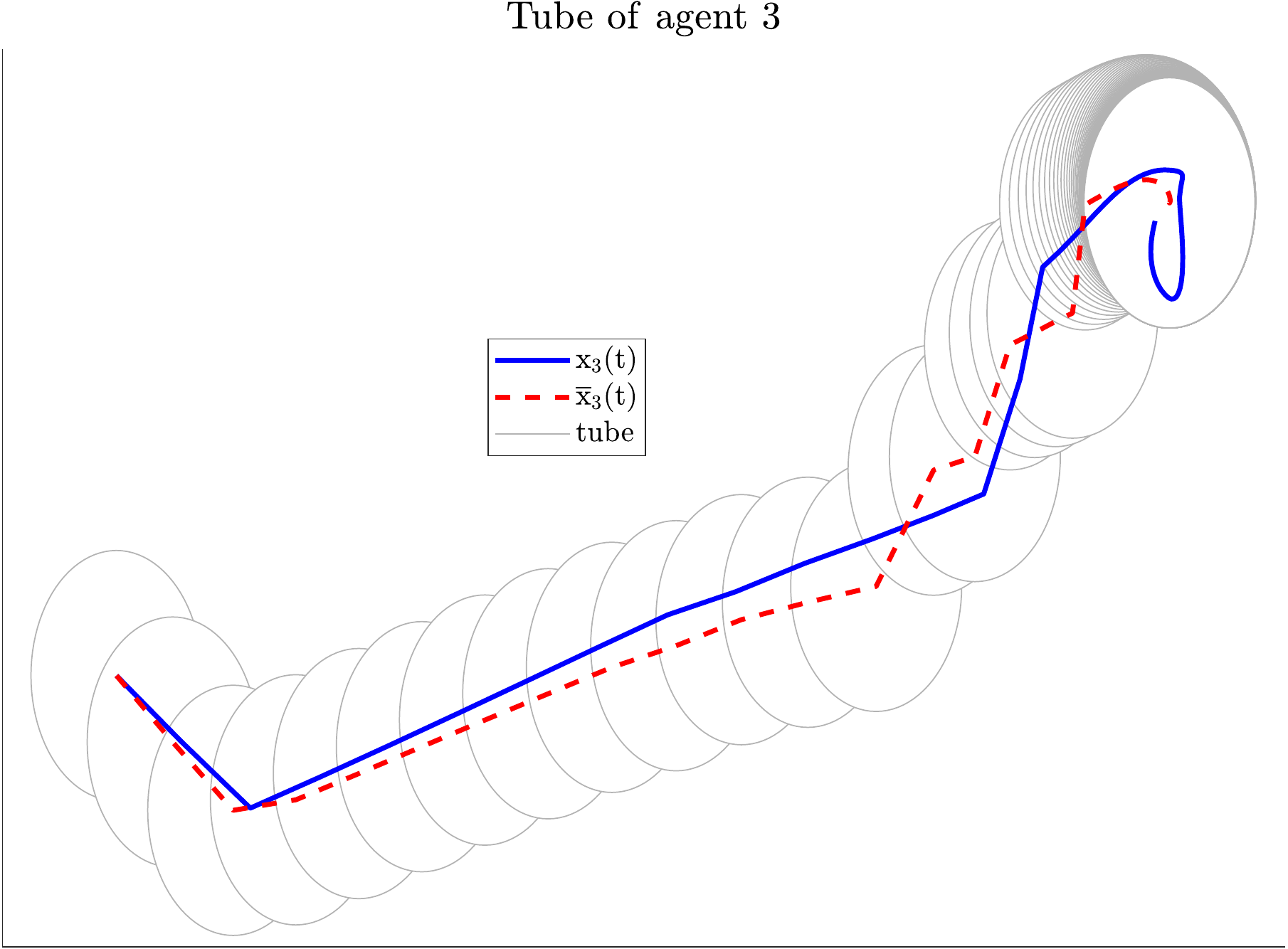}
	\caption{The evolution of the trajectory of agent $3$ in the workspace $\mathcal{D}$ over the time interval $[0,10] \sec$. The solid and the dashed lines represent the real and the nominal trajectory, respectively. The gray circles represent the tube.}\label{fig:tube3}
\end{figure}

\begin{figure}[t!]
	\centering
	\includegraphics[scale = 0.55]{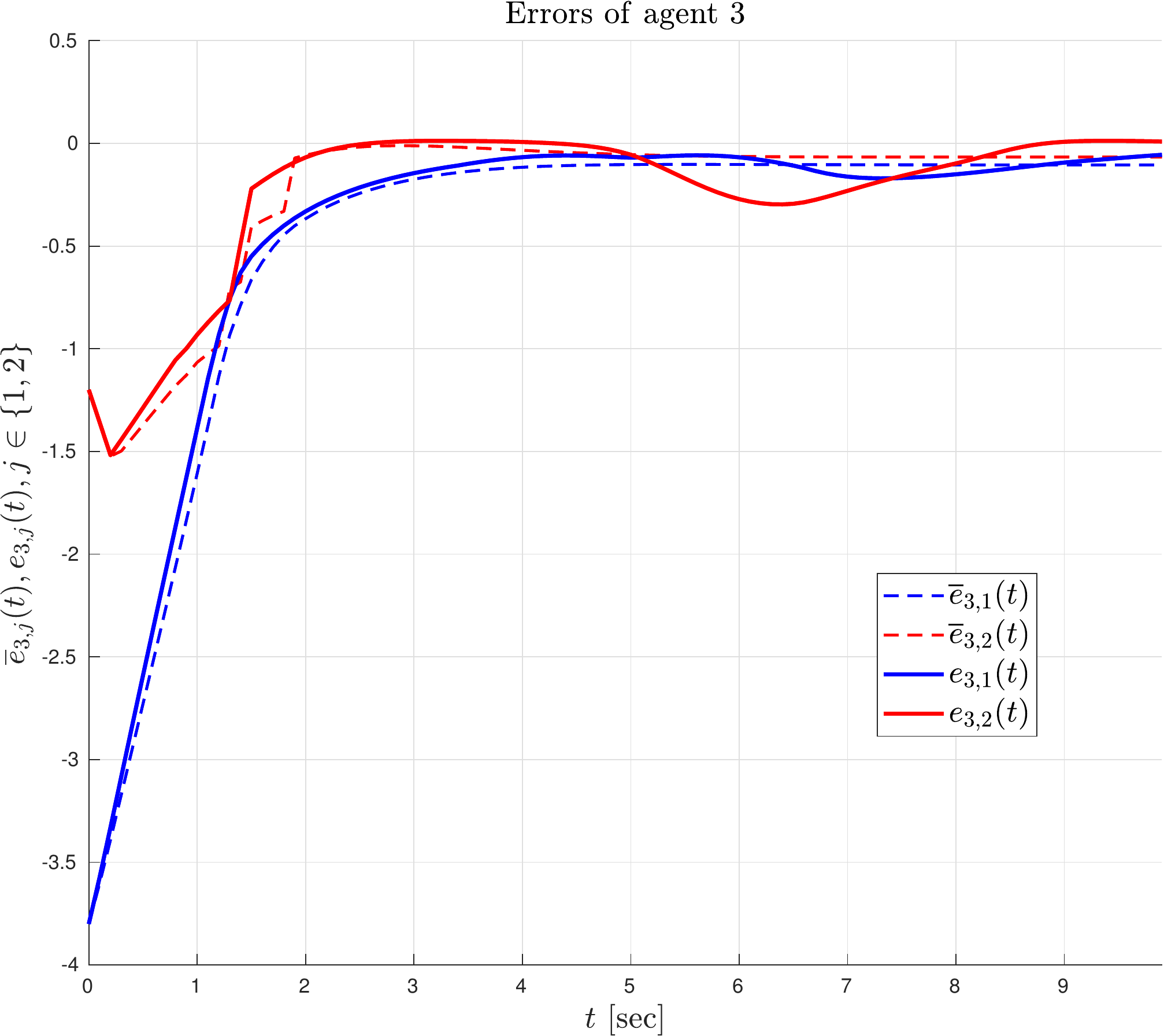}
	\caption{The evolution of the real error signals $e_3(t)$ and the nominal error signals $\overline{e}_3(t)$ over the time interval $[0,10] \sec$.}\label{fig:error3}
\end{figure}

\begin{figure}[t!]
	\centering
	\includegraphics[scale = 0.55]{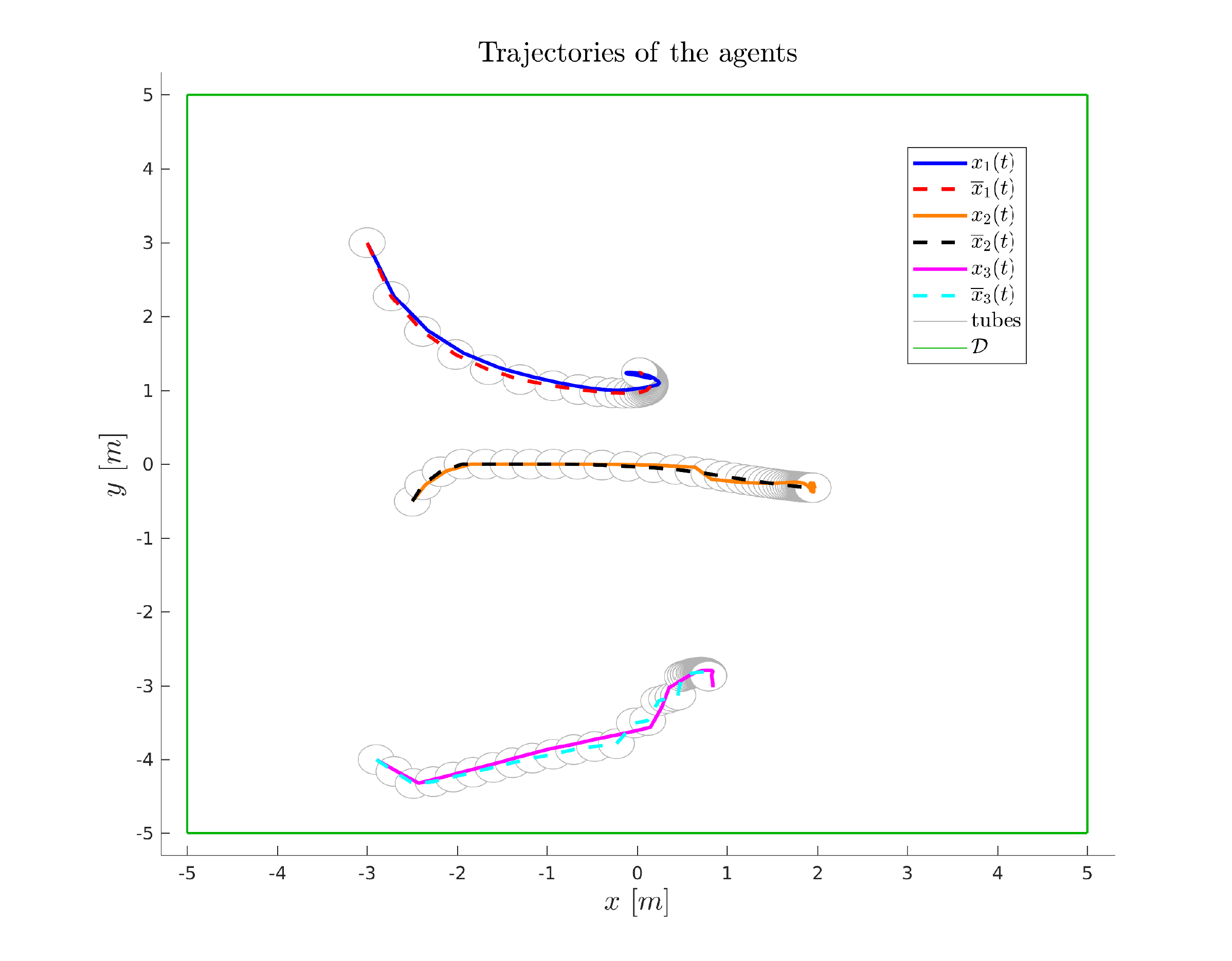}
	\caption{The workspace $\mathcal{D}$ along with the trajectories of all agents.}\label{fig:workspace}
\end{figure}

\begin{figure}[t!]
	\centering
	\includegraphics[scale = 0.55]{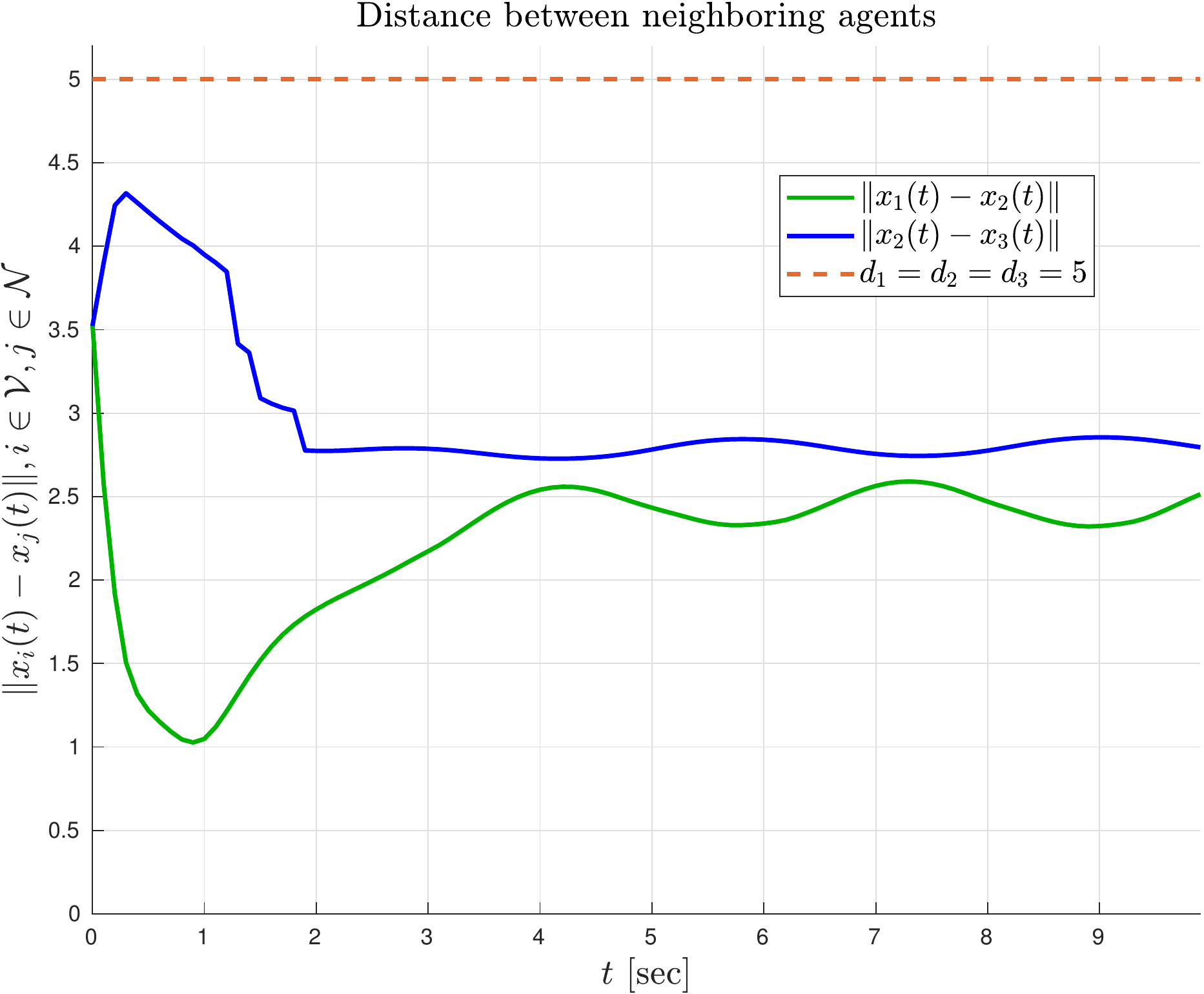}
	\caption{The distance between the neighboring agents $1-2$ and $2-3$. The distance remains below the threshold $d_1 = d_2 = d_3 = 5$ for all times, i.e., the connectivity of the neighboring agents is preserved for all times.}\label{fig:distance}
\end{figure}

\section{Conclusions and Future Research} \label{sec:conclusions}

This paper investigates the problem of decentralized tube-based MPC for uncertain nonlinear continuous-time multi-agent systems. Each agent has a limited sensing range within which can exchange information with neighboring agents. The task involves navigation to predefined configurations with connectivity preservation of the initially connected agents. Each agent solves a nominal DFHOCP in order to calculate a potion of its control input. The other portion of the control input is calculated offline in order to guarantee that the real trajectory of the agent remains in a bounded hyper-tube for all times, due to disturbances. Simulation results verify the proposed approach. Future efforts will be devoted towards the directions of reducing the communication burden between the agents by introducing event-triggered communication controllers.

\clearpage

{\centering \section*{Appendices}}

\appendix

\begin{figure}[t!]
	\centering
	\includegraphics[scale = 0.6]{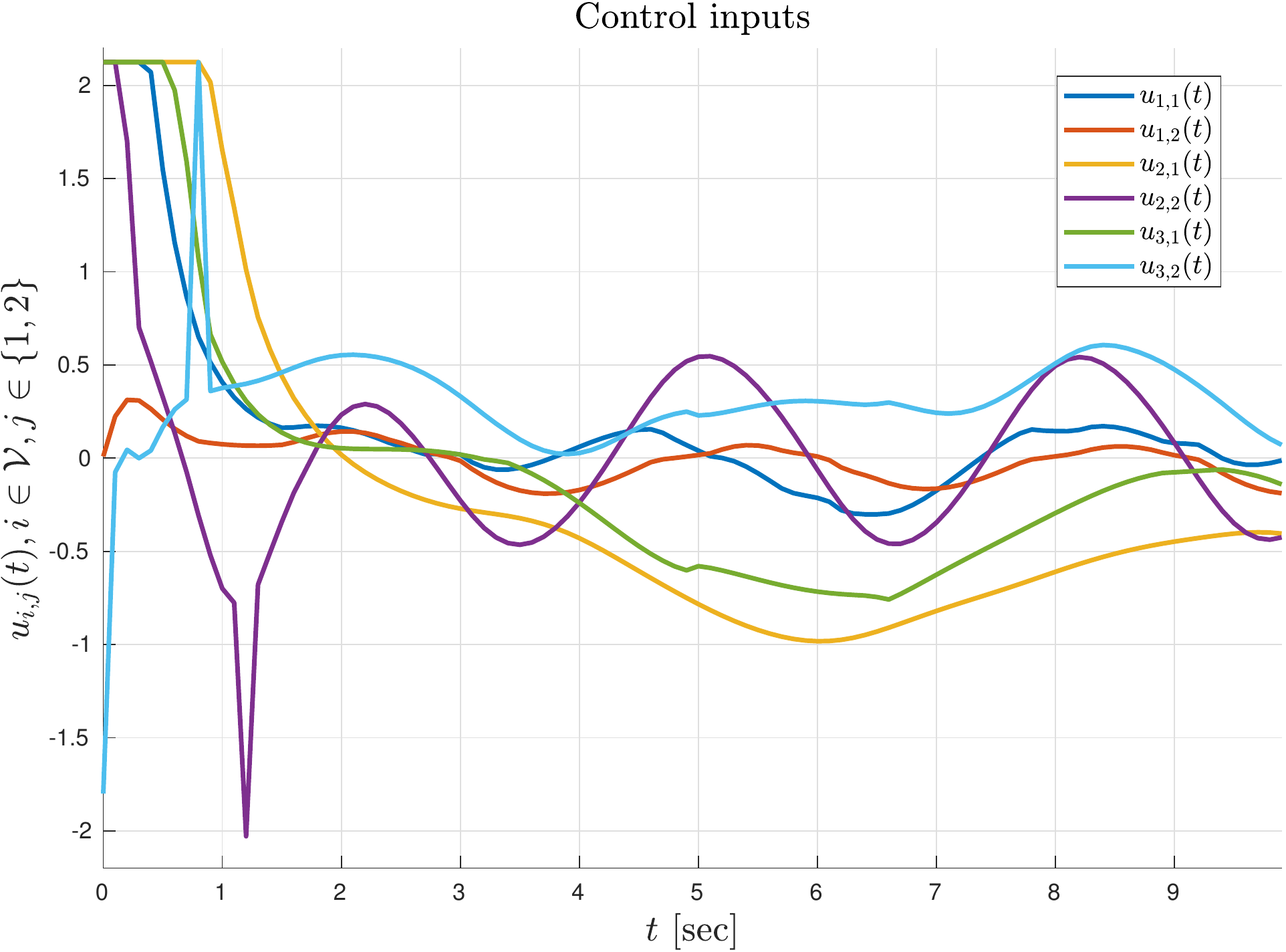}
	\caption{The control input signals $u_i(t)$, $i \in \mathcal{V}$ satisfying the constraints $u_i(t) \in \mathcal{U}_i$, $\forall i \in \mathcal{V}$, $t \in \mathbb{R}_{\ge 0}$.}\label{fig:inputs}
\end{figure}

\section{Proof of Lemma \ref{leamm:proposed_RPI_set}} \label{appendix:proof_lemma_proposed_RPI_set}

\noindent Consider the candidate Lyapunov function $\varphi(z_i) = \frac{1}{2} \|z_i\|^2 > 0$ with $\varphi(0) = 0$. The time derivative of $\varphi$ along the trajectories of the system \eqref{eq:z_dynamics}, is given by:
\begin{align*}
\dot{\varphi}(z_i) & = z_{i}^\top \dot{z}_{i} \notag \\
& = z_{i}^\top \Big[ \Lambda_i(e_i, \overline{e}_i, u_i)+f_i(\overline{e}_i+x_{\scriptscriptstyle i, \rm des}, u_i)- f_i(\overline{e}_i+x_{\scriptscriptstyle i, \rm des}, \overline{u}_i)+w_i\Big] \notag \\
& \le \|z_i\|\|\Lambda_i(e_i, \overline{e}_i, u_i)\|+z_{i}^\top w_i + z_{i}^\top \big[ f_i(\overline{e}_i+x_{\scriptscriptstyle i, \rm des}, u_i) - f_i(\overline{e}_i+x_{\scriptscriptstyle i, \rm des}, \overline{u}_i) \big], 
\end{align*}
which by using the upper bound of $\Lambda_i$ as in \eqref{eq:L_bound}, becomes:
\begin{align}
\dot{\varphi}(z_i) & \le L_i \|z_i\|^2+z_{i}^\top w_i + z_{i}^\top \big[ f_i(\overline{e}_i+x_{\scriptscriptstyle i, \rm des}, u_i) - f_i(\overline{e}_i+x_{\scriptscriptstyle i, \rm des}, \overline{u}_i) \big].  \label{eq:lyap1}
\end{align}
By invoking Lemma \ref{lemma:quadr_forms}, for $M = I_n$ we have:
\begin{align*}
z_{i}^\top w_i & \le \frac{\|z_{i}\|^2}{4 \rho_i}+\rho_i\|w_i\|^2 \le \frac{\|z_{i}\|^2}{4 \rho_i}+ \rho_i \widetilde{w}_i^2,
\end{align*}
for any constants $\rho_i > 0$. By using the aforementioned inequalities, \eqref{eq:lyap1} becomes:
\begin{align}
\dot{\varphi}(z_i) & \le \left(L_i+\frac{1}{4 \rho_i}\right) \|z_i\|^2+ \rho_i \widetilde{w}_i^2 + z_{i}^\top \big[ f_i(\overline{e}_i+x_{\scriptscriptstyle i, \rm des}, u_i) - f_i(\overline{e}_i+x_{\scriptscriptstyle i, \rm des}, \overline{u}_i) \big]. \label{eq:lyap_2}
\end{align}
According to Proposition \ref{eq:MVT}, and due to the fact that the sets $\mathcal{U}_i$ are convex, i.e., $${\rm Co}(u_i, \overline{u}_i) \subseteq \mathcal{U}_i, \forall u_i, \overline{u}_i \in \mathcal{U}_i, i \in \mathcal{V},$$ there exist constant vectors $\xi_{i,1}$, $\dots,$ $\xi_{i,n} \in {\rm Co}(u_i,\overline{u}_i)$ such that:
\begin{align*}
f_i(\overline{e}_i+x_{\scriptscriptstyle i, \rm des}, u_i)-f_i(\overline{e}_i+x_{\scriptscriptstyle i, \rm des}, \overline{u}_i)  & = \left[ \sum_{k = 1}^{n} \sum_{j=1}^{n} \ell_n(k) \ell_n(j)^\top \frac{\partial f_{i,k}(\overline{e}_i+x_{\scriptscriptstyle i, \rm des}, \xi_{i,k})}{\partial u_j} \right] (u_i-\overline{u}_i) \notag \\
& = J_i(\overline{e}_i+x_{\scriptscriptstyle i, \rm des}, \xi_{i,k}) (u_i - \overline{u}_i).
\end{align*}
Then, \eqref{eq:lyap_2} becomes:
\begin{align}
\dot{\varphi}(z_i) & \le \left(L_i+\frac{1}{4 \rho_i}\right) \|z_i\|^2+\rho_i\widetilde{w}_i^2 + z_{i}^\top J_i(\overline{e}_i+x_{\scriptscriptstyle i, \rm des}, \xi_{i,k}) (u_i - \overline{u}_i). 
\end{align}
\noindent By designing the control laws $u_i - \overline{u}_i$ as in \eqref{eq:control_law_u_kappa}, \eqref{eq:kappa_law} we get:
\begin{align*}
\dot{\varphi}(z_i) & \le \left(L_i+\frac{1}{4 \rho_i}\right) \|z_i\|^2+\rho_i \widetilde{w}_i^2 - k_i z_{i}^\top J_i(\overline{e}_i+x_{\scriptscriptstyle i, \rm des}, \xi_{i,k}) z_{i}. 
\end{align*}
Writing the matrices $J_i$ as $J_i = \frac{J_i+J_i^\top}{2}+\frac{J_i-J_i^\top}{2}$ and taking into account that $y^\top\left(\frac{J_i-J_i^\top}{2}\right) y = 0$, $\forall y \in \mathbb{R}^n$ we get:
\begin{align*}
\dot{\varphi}(z_i) & \le \left(L_i+\frac{1}{4 \rho_i}\right) \|z_i\|^2+\rho_i \widetilde{w}_i^2 - k_i z_{i}^\top \left[\frac{J_i(\overline{e}_i+x_{\scriptscriptstyle i, \rm des}, \xi_{i,k})+J_i^\top(\overline{e}_i+x_{\scriptscriptstyle i, \rm des}, \xi_{i,k})}{2}\right] z_{i}. 
\end{align*}
By using \eqref{eq:lambda_min} from Assumption \ref{ass:lower_bound_deriv} and the fact that:
\begin{align*}
y^\top P y \ge \lambda_{\min}(P)\|y\|^2, \forall y \in \mathbb{R}^n, P \in \mathbb{R}^{n \times n}, P >0,
\end{align*}
we obtain:
\begin{align*}
\dot{\varphi}(z_i) & \le \left(L_i+\frac{1}{4 \rho_i}\right) \|z_i\|^2+ \rho_i \widetilde{w}_i^2 - k_i \lambda_{\min} \left[ \frac{J_i(\overline{e}_i+x_{\scriptscriptstyle i, \rm des}, \xi_{i,k})+J^\top_i(\overline{e}_i+x_{\scriptscriptstyle i, \rm des}, \xi_{i,k})}{2}\right] \|z_{i}\|^2 \notag \\
& \le \left(L_i+\frac{1}{4 \rho_i}\right) \|z_i\|^2+\rho_i \widetilde{w}_i^2 - k_i \underline{J}_i \|z_{i}\|^2 \notag \\
& = -\left(k_i \underline{J}_i - L_i-\frac{1}{4 \rho_i}\right) \|z_i\|^2+ \rho_i \widetilde{w}_i^2 \notag \\
& = - \underline{J}_i \left[k_i - \frac{1}{\underline{J}_i} \left(L_i + \frac{1}{4 \rho_i}\right) \right] \|z_i\|^2+ \rho_i \widetilde{w}_i^2. 
\end{align*}
By designing the control gains $k_i$ as in \eqref{eq:control_gains2}, it yields:
\begin{align*}
\dot{\varphi}(z_i) & = -\underline{k}_i \underline{J}_i \|z_i\|^2+ \rho_i \widetilde{w}_i^2.
\end{align*}
Thus, it is guaranteed that $\dot{\varphi}(z_i) < 0$ when: $\|z_i\| > \widetilde{z}_i$, where $\widetilde{z}_i$ is given in \eqref{eq:Omega_set}. By invoking Theorem \ref{teheorem:uub_theorem} and due to the fact that $z_i(0) = 0$, $\forall i \in \mathcal{V}$ we have that:
\begin{align*}
\|z_i(t)\| \le \widetilde{z}_i, \forall t \in \mathbb{R}_{\ge 0}, i \in \mathcal{V},
\end{align*}
which leads to the conclusion of the proof.
\qed

\section{Feasibility Analysis of Theorem \ref{theorem}} \label{app:feasibility_analysis}

Consider a sampling instant $t_k$ for which a solution $\overline{u}_i^{\star}\big(\cdot;\ \bar{e}_i(t_k)\big)$ to DFHOCP \eqref{eq:mpc_cost_function}-\eqref{eq:mpc_terminal_set} of agent $i \in \mathcal{V}$ exists. Suppose now a time instant $t_{k+1}$ such that $t_{k+1} = t_k + \delta$, and consider that the optimal control signal calculated at $t_k$ is comprised of the following two portions:
\begin{equation} \label{eq:optimal_input_portions}
\overline{u}_i^{\star}\big(\cdot;\ \overline{e}_i(t_k)\big) = \left\{
\begin{array}{ll}
\overline{u}_i^{\star}\big(s;\ \overline{e}_i(t_k)\big), & s \in [t_k, t_{k+1}] \\
\overline{u}_i^{\star}\big(s;\ \overline{e}_i(t_k)\big), & s \in [t_{k+1}, t_k + T_p]
\end{array},
\right.
\end{equation}

Both portions are admissible since the calculated optimal control input is admissible, and hence they both conform to the input constraints. Furthermore, the predicted states $\overline{e}_i(s; \ \overline{u}_i^\star(\cdot), \overline{e}_i(t_k))$ will satisfy the state constraints for every $s \in [t_k, t_k+T]$ and it also holds that:
\begin{align}
\overline{e}_i\big(t_k + T;\ \overline{u}_i^{\star}(\cdot), \overline{e}_i(t_k)\big) \in \mathcal{F}_i.
\end{align}

As discussed in Section \ref{sec:online_control_design}, according to Assumption \ref{ass:stabilized_linear_assumption}, there exists an admissible control input $u_{i, \text{loc}}(\overline{e}_i)$ that renders $\mathcal{F}_i$ invariant over $[t_k + T, t_{k+1} + T]$. Given the above facts, we can construct an admissible input $\overline{u}_i(\cdot)$ starting at time $t_{k+1}$ by sewing together the second portion of \eqref{eq:optimal_input_portions} and the input $u_{i, \text{loc}}(\overline{x})$ as:
\begin{equation*}
\widetilde{u}_i(s) = \left\{
\begin{array}{ll}
\overline{u}_i^{\star}\big(s;\ \overline{e}_i(t_k)\big), & s \in [t_{k+1}, t_k + T] \\
u_{i, \text{loc}}\big(\overline{e}_i(s)\big), & s \in (t_k + T, t_{k+1} + T]\\
\end{array},
\right.
\end{equation*}
Applied at time $t_{k+1}$, $\widetilde{u}_i(s)$ is an admissible control input with respect to the input constraints as a composition of admissible control inputs, for all $s \in [t_{k+1}, t_{k+1}+T]$. What remains to prove is the following statement.

\noindent \textbf{Statement} :  $e_i\big(t_{k+1} + s;\ \overline{u}_i^{\star}(\cdot), e_i(t_{k+1})\big) \in \mathcal{E}_i$, $\forall s \in [0,T]$.

\noindent Initially, at time $t_{k+1}$, $\widetilde{u}_i$ is an admissible control input according to Definition \ref{definition:admissible_input_with_disturbance}. The, according to $3)$ of Definition \ref{definition:admissible_input_with_disturbance} we have that:
\begin{align*}
& \overline{e}_i\big(t_{k+1} + s;\ \overline{u}_i^{\star}(\cdot), \overline{e}_i(t_{k+1})\big) \in \overline{\mathcal{E}_i} = \mathcal{E}_i \ominus \mathcal{Z}_i, \forall s \in [0, T], \notag
\end{align*}
By invoking \eqref{eq:error_z} and the fact that $\mathcal{Z}_i$ are RCI sets it is guaranteed that: 
\begin{align*}
& e_i\big(t_{k+1} + s;\ \overline{u}_i^{\star}(\cdot), e_i(t_{k+1})\big)- \overline{e}_i\big(t_{k+1} + s;\ \overline{u}_i^{\star}(\cdot), \overline{e}_i(t_{k+1})\big) \in \mathcal{Z}_i, \forall s \in [0,T].
\end{align*}
Adding the latter to the former yields:
\begin{align*}
e_i\big(t_{k+1} + s;\ \overline{u}_i^{\star}(\cdot), e_i(t_{k+1})\big) \in \big( \mathcal{E}_i \ominus \mathcal{Z}_i \big) \oplus \mathcal{Z}_i, \forall s \in [0, T].
\end{align*}
By using Property \ref{prop:set_prop} we have that: $\big( \mathcal{E}_i \ominus \mathcal{Z}_i \big) \oplus \mathcal{Z}_i \subseteq \mathcal{E}_i$. Thus, it holds that $e_i\big(t_{k+1} + s;\ \overline{u}_i^{\star}(\cdot)$, $e_i(t_{k+1})\big) \in \mathcal{E}_i, \forall s \in [0, T]$, which concludes the proof of the statement. \\ \\ \indent By taking the aforementioned into consideration, the feasibility of a solution to the optimization problem at time $t_k$ implies feasibility at all times $t_{n+1}$, with $n > k$. Thus, since at time $t=0$ a solution is assumed to be feasible, a solution to the optimal control problem is feasible for all $t \in \mathbb{R}_{\ge 0}$, and for all agents $i \in \mathcal{V}$. \qed

\section{Convergence Analysis of Theorem \ref{theorem}} \label{app:convergence_analysis}

Due to the fact the sets $\mathcal{Z}_i$ as given in \eqref{eq:Omega_set}, are RCI, i.e.,  $z_i(t) \in \mathcal{Z}_i$, for every $t \in \mathbb{R}_{\ge 0}$, there exist class $\mathcal{K}$ functions $\gamma_i$ (see \cite[Sec. 4.9, p. 175]{khalil_nonlinear_systems}) such that:
\begin{equation} \label{eq:ISS_arg2}
\|z_i(t)\| \le \gamma_i \left(\displaystyle \sup_{0 \le \tau \le t} \|w_i(\tau)\|\right), \forall t \in \mathbb{R}_{\ge 0}, i \in \mathcal{V}.
\end{equation}

Since only the nominal system dynamics \eqref{eq:nominal_system} are used for the online computation of the control actions $\overline{u}_i(s) \in \overline{\mathcal{U}}_i$, $s \in [t_k, t_k+T]$ through the DFHOCP \eqref{eq:mpc_cost_function}-\eqref{eq:mpc_terminal_set}, by invoking nominal NMPC stability results found on \cite{frank_1998_quasi_infinite, nmpc_bible}, it can be proven that the NMPC control law $\overline{u}_i$ renders the closed loop trajectories of the nominal system \eqref{eq:nominal_system} asymptotically ultimated bounded in the sets $\mathcal{F}_i$, for all $i \in \mathcal{V}$, as given in Definition \ref{def:asympt_ultimately_bounded}. Then, from \cite[Lemma 4.5, p. 150]{khalil_nonlinear_systems}, there exist class $\mathcal{KL}$ functions $\beta_i$, such that:
\begin{equation} \label{eq:ISS_arg1}
\|\overline{e}_i(t)\| \le \beta_i(\|\overline{e}_i(0)\|, t), \forall t \in \mathbb{R}_{\ge 0}, i \in \mathcal{V}.
\end{equation}
According to \eqref{eq:error_z} it holds that: 
\begin{align*}
e_i(t) & = \overline{e}_i(t)+z_i(t) \notag \\
\Rightarrow \|e_i(t)\| & = \|\overline{e}_i(t)+z_i(t)\|  \le  \|\overline{e}_i(t)\| + \|z_i(t)\|, i \in \mathcal{V},
\end{align*}
and $e_i(0) =  \overline{e}_i(0)$. By combining the latter with \eqref{eq:ISS_arg2} and \eqref{eq:ISS_arg1}, we conclude that for all agents $i \in \mathcal{V}$ and for all the conditions $e_i(0) \in \mathcal{E}_i$ it holds that:
\begin{equation*}
\|e_i(t)\| \leq \beta_i \big(\|e_i(0)\|,t\big) + \gamma_i \left(\displaystyle \sup_{0 \le \tau \le t} \|w_i(\tau)\|\right), \forall t \in \mathbb{R}_{\ge 0}.
\end{equation*}
Thus, we have shown that the proposed control law \eqref{eq:control_law_u_kappa} renders the closed-loop system \eqref{eq:closed_loop_system} ISS with reference to the disturbances $w_i(t) \in \mathcal{W}_i$, for every initial condition $x_i(0) \in \mathcal{X}_i$ and $i \in \mathcal{V}$. \qed

\clearpage

\bibliography{references}
\bibliographystyle{ieeetr}

\clearpage 

\begin{wrapfigure}{l}{25mm} 
\includegraphics[width=1in,height=1.25in,clip,keepaspectratio]{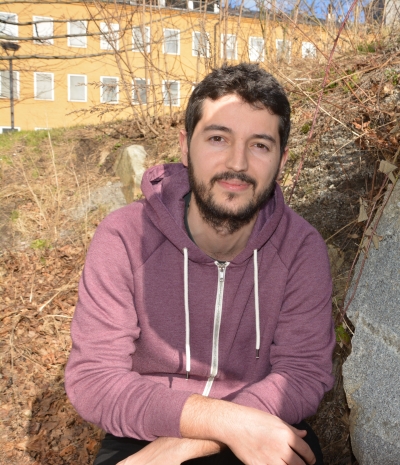}
\end{wrapfigure}\par
\textbf{Alexandros Nikou} was born in Athens, Greece, in 1988. He received the Diploma in Electrical and Computer Engineering in 2012 and the M.Sc. in Automatic Control in 2014, both from National Technical University of Athens (NTUA), Greece. He is currently a PhD student at the Department of Automatic Control, School of Electrical Engineering and Computer Science, KTH Royal Institute of Technology, Stockholm, Sweden. His current research interests include Multi-Agent Systems Control, Distributed Nonlinear Model Predictive Control and Formal Methods in Control.\par

\vspace{5mm}

\begin{wrapfigure}{l}{25mm} 
\includegraphics[width=1.6in,height=1.25in,clip,keepaspectratio]{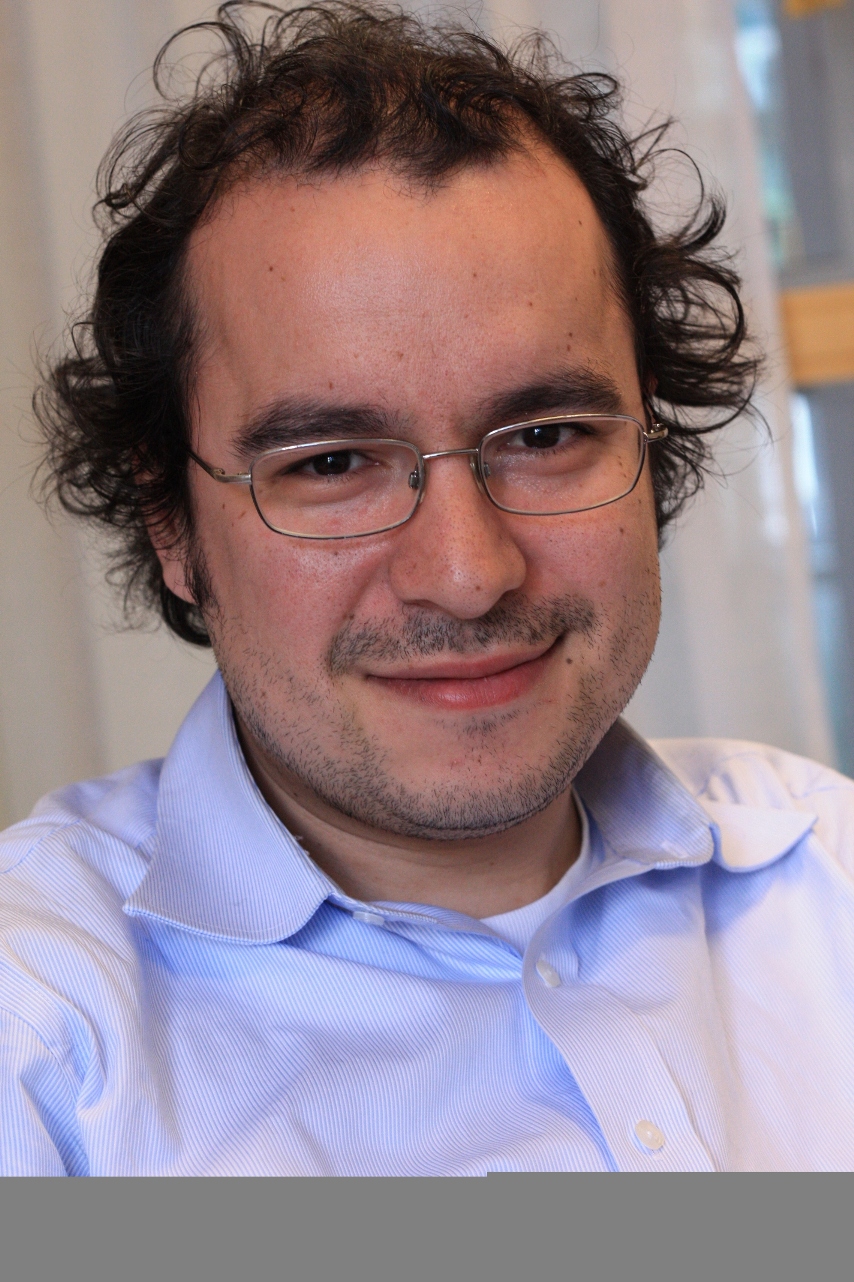}
\end{wrapfigure}\par
\textbf{Dimos Dimarogonas} was born in Athens, Greece, in 1978. He received the Diploma in Electrical and Computer Engineering in 2001 and the Ph.D. in Mechanical Engineering in 2007, both from National Technical University of Athens (NTUA), Greece. Between May 2007 and February 2009, he was a Postdoctoral Researcher at the Department of Automatic Control, School of Electrical Engineering and Computer Science, Royal Institute of Technology (KTH), Stockholm, Sweden. Between February 2009 and March 2010, he was a Postdoctoral Associate at the Laboratory for Information and Decision Systems (LIDS) at the Massachusetts Institute of Technology (MIT), Boston, MA, USA. He is currently Professor at the Department of Automatic Control, KTH Royal Institute of Technology, Stockholm, Sweden. His current research interests include Multi-Agent Systems, Hybrid Systems and Control, Robot Navigation and Networked Control. He serves in the Editorial Board of Automatica, the IEEE Transactions on Automation Science and Engineering and the IET Control Theory and Applications and is a Senior member of IEEE and the Technical Chamber of Greece.\par
\end{document}